\documentclass[12pt] {article}
\usepackage{amsmath,amssymb,amsfonts}
\usepackage{amsthm}
\usepackage{bm}
\usepackage[all]{xy}
\usepackage{tikz-cd}
\usepackage[a4paper, top=2.5cm, bottom=2.5cm, left=2.5cm, right=2.5cm]{geometry}
\usepackage{url} 
\usepackage[colorlinks=true, urlcolor=blue, linkcolor=blue]{hyperref}

\newtheorem{definition}{Definition}[section]
\newtheorem{theorem}{Theorem}[section]
\newtheorem{lemma}{Lemma}[section]

\newtheorem{proposi}{Proposition}[section]
\newtheorem{remark}{Remark}[section]
\usepackage{titlesec}
\titleformat{\subsection}[runin]{\normalfont\bfseries}{\thesubsection}{1em}{}
\titlespacing*{\subsection}{\parindent}{0pt}{1em} 
\titleformat{\subsubsection}[runin]{\normalfont\bfseries}{\thesubsubsection}{1em}{}
\titlespacing*{\subsubsection}{\parindent}{0pt}{1em}
\title{A Simplicial Approach to Higher Geometric Quantization}
\author{Zhang Qian}
\date{}
\begin{document}
\maketitle

\begin{abstract}
	This paper develops a unified framework for observables in $n$-plectic geometry, extending the $L_\infty$-algebra of Hamiltonian $(n-1)$-forms to Hamiltonian forms of all degrees via a degree‑shifting Grassmann variable $u$ that encodes submanifold codimension. Interpreting $k$-form observables as $k$-dimensional topological defects yields a recursive gluing construction that assembles into a semi‑simplicial set $\bm{\mathrm{sOb}}_\bullet(M)$, which we prove satisfies the Kan filling property, thereby providing an $n$-groupoid model for observables. From this semi‑simplicial perspective we extract cohomological invariants and construct a recursive inner product leading to a categorified pre‑$n$-Hilbert space. The hierarchical structure of polarizations yields a natural quantization scheme matching the $1$-polarization classification of multisymplectic geometry. The resulting framework bridges higher algebraic structures with higher categorical geometry and establishes a systematic foundation for the geometric quantization of extended objects.
\end{abstract}
\tableofcontents
\section{Introduction}\label{s1}

Multisymplectic geometry generalizes symplectic geometry to manifolds $M$ equipped with a closed, nondegenerate $(n+1)$-form $\omega$; the pair $(M,\omega)$ is called an $n$-plectic manifold~\cite{HF,RN}. While symplectic geometry ($n=1$) describes point particles, $n$-plectic geometry governs the dynamics of $(n-1)$-dimensional extended objects~\cite{Baez2010ya,Baez2005qu}. The coupling of such objects to background fields is encoded by Wess--Zumino terms, generalizing the minimal coupling of a point particle to an electromagnetic field~\cite{G1983w,G1983,Witten1983ar}. Moreover, multisymplectic methods retain manifest covariance, avoiding the $3+1$ split of canonical quantization, and therefore provide a natural framework for field theory.

In $n$-plectic geometry, observables are Hamiltonian forms $\alpha$ satisfying $d\alpha = -\iota_{v_\alpha}\omega$ for some multivector field $v_\alpha$. Seminal work by Baez, Rogers and others established that the space of top-degree Hamiltonian forms (degree $n-1$) carries an $L_\infty$-algebra structure, generalising the Poisson algebra of symplectic geometry~\cite{Baez2009,Rogers_2011}. Restricting to top-degree forms, however, is both mathematically unnatural and physically inadequate: observables of various degrees appear naturally in extended topological field theories~\cite{Baez1995,Lurie2009}, generalised symmetries~\cite{B2023,B2024,B2025,G2014}, and topological phases of matter~\cite{S2018,Q2024}; in our framework such observables are precisely Hamiltonian forms, with a Hamiltonian $k$-form interpreted as a $k$-dimensional observable. A complete theory of observables in multisymplectic geometry must therefore encompass Hamiltonian forms of every degree.

The present work addresses this need and simultaneously lays the combinatorial groundwork for higher geometric quantization. We proceed through the following constructions.

\noindent\textbf{Algebraic construction.}
We extend the $L_\infty$-algebra of observables to all Hamiltonian forms by introducing a formal variable $u$ of bidegree $(-1,1)$. This variable shifts the form degree and encodes the codimension of a submanifold inside the worldvolume of an extended object. The resulting bigraded space $\bigoplus_{j=0}^{n-1}\Omega_{\mathrm{Ham}}^{n-j-1}(M)\langle u^j\rangle$ is equipped with higher brackets defined via interior products with Hamiltonian multivector fields. We prove that these brackets satisfy the homotopy Jacobi identities, yielding a graded $L_\infty$-algebra (Theorem~\ref{thmmain}) that unifies observables on submanifolds of different codimensions. Within this algebra we identify a higher Heisenberg subalgebra generated by mutually commuting Hamiltonian vector fields, which generalizes the classical Heisenberg algebra of symplectic geometry.

\noindent\textbf{Geometric construction.}
Interpreting $k$-form Hamiltonian observables as $k$-dimensional topological defects, we develop a recursive gluing construction that assigns Hamiltonian forms to submanifolds of every codimension. This structure naturally assembles into a semi-simplicial set $\bm{\mathrm{sOb}}_\bullet(M)$ (Definition~\ref{s411}), whose $k$-simplices are smooth singular $k$-simplices equipped with Hamiltonian $k$-forms. We prove that $\bm{\mathrm{sOb}}_\bullet(M)$ satisfies the Kan filling property (Theorem~\ref{s4t1}), thereby providing an $n$-groupoid model for observables. The linearization $\mathbb{R}[\bm{\mathrm{sOb}}_\bullet(M)]$ yields a combinatorial $n$-vector space, and morphisms are captured by the path-space semi-simplicial set $\bm{\mathrm{sOb}}_\bullet(\bm{\mathrm{DiffMan}}(\Delta^1,M))\cong\bm{\mathrm{sOb}}_{\bullet+1}(M)$. A cochain complex extracted from $\bm{\mathrm{sOb}}_\bullet(M)$ encodes gluing obstructions and adiabatic invariants.

\noindent\textbf{Higher geometric quantization.}
The framework developed in this paper also lays the combinatorial foundation for the geometric quantization of $n$-plectic manifolds. Under the prequantum condition $[\omega/2\pi]\in H^{n+1}(M,\mathbb{Z})$, the semi-simplicial set $\bm{\mathrm{sOb}}_\bullet(M)$ furnishes a combinatorial model of the corresponding $n$-gerbe with connection. Quantum states are defined as $U(1)$-valued cochains on $\bm{\mathrm{sOb}}_\bullet(M)$, forming a cosimplicial set $\bm{\mathrm{qOb}}^\bullet(M)$, and a recursive inner product taking values in $U(1)[\bm{\mathrm{qOb}}^{\bullet+1}(M)]$ endows the linearized space with the structure of a categorified pre-$n$-Hilbert space. The resulting hierarchy of polarizations matches the $1$-polarization scheme of multisymplectic geometry~\cite{r2011} and offers a systematic path toward higher quantisation. 

The paper is organized as follows. Section~\ref{s2} recalls the necessary background on graded linear algebra, multivector calculus, $L_\infty$-algebras, simplicial sets, and $n$-plectic geometry. Section~\ref{s3} contains the algebraic construction of the graded $L_\infty$-algebra of all Hamiltonian forms and the higher Heisenberg subalgebra. Section~\ref{s4} develops the geometric framework: the semi-simplicial set $\bm{\mathrm{sOb}}_\bullet(M)$, its Kan property, the associated $n$-vector space, and cohomological invariants. Section~\ref{s5} relates $\bm{\mathrm{sOb}}_\bullet(M)$ to $n$-gerbes and prequantization, introduces the cosimplicial set of quantum states, the recursive inner product, and the hierarchy of polarizations. Section~\ref{s6} concludes with an outlook on future directions.

\section{Preliminaries and Notation convention}\label{s2}
One of the main objectives of this work is to demonstrate that the algebra of observables in multisymplectic geometry carries a natural $L_\infty$-algebra structure. The purpose of this section is twofold: to set up the necessary foundations and to fix the notation used throughout.

\subsection{Graded linear algebra}
Let $V$ be a $\mathbb{Z}$-graded vector space. For homogeneous elements $x_1, \dots, x_n \in V$ and a permutation $\sigma \in S_n$, the Koszul sign $\epsilon(\sigma) = \epsilon(\sigma; x_1, \dots, x_n)$ is defined by
$$x_1 \wedge\cdots\wedge x_n=\epsilon\left(\sigma;x_1,\cdots,x_n\right)x_{\sigma\left(1\right)}\wedge\cdots\wedge x_{\sigma\left(n\right)}$$
in the free graded-commutative algebra generated by $V$. Let $(-1)^\sigma$ denote the usual permutation sign; note that $\epsilon(\sigma)$ excludes this factor.

A permutation $\sigma \in S_{p+q}$ is a $(p,q)$-unshuffle if $\sigma(i) < \sigma(i+1)$ whenever $i = p$. The set of such unshuffles is denoted $\operatorname{Sh}(p,q)$ (e.g., $\operatorname{Sh}(2,1) = {\mathrm{id}, (2,3), (1,2,3)}$).

A linear map $f: V^{\otimes n} \to W$ between graded vector spaces is graded skew-symmetric iff
$$f\left(v_{\sigma\left(1\right)},\cdots,v_{\sigma\left(n\right)}\right)=\left(-1\right)^{\sigma}\epsilon\left(\sigma\right)f\left(v_1,\cdots,v_n\right),$$
for all $\sigma \in S_n$. The degree of $x_1 \otimes \cdots \otimes x_n$ is $|x_1 \otimes \cdots \otimes x_n| = \sum_{i=1}^{n} |x_i|$.

\subsection{Multivector calculus}
Let $\mathfrak{X}(M)$ denote the $C^\infty(M)$-module of vector fields on a smooth manifold $M$. The graded-commutative algebra of multivector fields is defined as $\mathfrak{X}^{\wedge\bullet}(M) = \bigoplus_{k=0}^{\dim M} \bigwedge^{k} \mathfrak{X}(M)$. This space carries the Schouten-Nijenhuis bracket $[\cdot,\cdot]$, a degree$-1$ Lie bracket satisfying the graded Leibniz rule, which makes $\mathfrak{X}^{\wedge\bullet}(M)$ into a Gerstenhaber algebra. For decomposable multivectors, the bracket is given explicitly by
\begin{align}
	[u_1\wedge\cdots\wedge u_m,\; v_1\wedge\cdots\wedge v_n] = &\sum_{i=1}^{m}\sum_{j=1}^{n} (-1)^{i+j}\,[u_i,v_j]\wedge u_1\wedge\cdots\wedge\hat{u}_i\wedge\cdots\wedge u_m \nonumber\\
	&\wedge v_1\wedge\cdots\wedge\hat{v}_j\wedge\cdots\wedge v_n,
\end{align}
where $[u_i,v_j]$ denotes the ordinary Lie bracket of vector fields and the hats indicate omitted factors.

The following proposition summarizes the fundamental properties of the Schouten-Nijenhuis bracket.

\begin{proposi}[\cite{MP}]\label{s2pro}
	For homogeneous multivector fields $u, v, w$, the following identities hold:
	\begin{itemize}
		\item[(1)] $[u,v] = -(-1)^{(|u|-1)(|v|-1)}[v,u]$;
		\item[(2)] $[u,[v,w]] = [[u,v],w] + (-1)^{(|u|-1)(|v|-1)}[v,[u,w]]$;
		\item[(3)] $[u,v\wedge w] = [u,v]\wedge w + (-1)^{(|u|-1)|v|}\,v\wedge[u,w]$,
		
		$[u\wedge v,w] = u\wedge[v,w] + (-1)^{(|w|-1)|v|}\,[u,w]\wedge v$.
	\end{itemize}
\end{proposi}

From Proposition~\ref{s2pro} we immediately obtain the following useful identity, which will be essential in computing the $L_\infty$-algebraic structure of observables.

\begin{lemma}\label{s2l1}
	For homogeneous multivector fields $u$ and $v_1,\dots,v_n$ on a manifold $M$,
	\begin{align*}
		[u, v_1\wedge\cdots\wedge v_n] = \sum_{i=1}^{n} (-1)^{\sum_{\alpha=1}^{i-1} |v_i|\,|v_\alpha|}\,
		[u,v_i]\wedge v_1\wedge\cdots\wedge \hat{v}_i\wedge\cdots\wedge v_n.
	\end{align*}
\end{lemma}

\begin{proof}
	We proceed by induction on $n$. The case $n=1$ is trivial. Assume the formula holds for $n=k-1$. For $n=k$,
	\begin{align*}
		[u, v_1\wedge\cdots\wedge v_k]
		&= [u, v_1\wedge\cdots\wedge v_{k-1}]\wedge v_k \\
		&\qquad + (-1)^{(|u|-1)\sum_{a=1}^{k-1}|v_a|}\,
		v_1\wedge\cdots\wedge v_{k-1}\wedge[u,v_k] \\[2mm]
		&= \sum_{i=1}^{k-1} (-1)^{\sum_{\alpha=1}^{i-1}|v_i||v_\alpha|}\,
		[u,v_i]\wedge v_1\wedge\cdots\wedge \hat{v}_i\wedge\cdots\wedge v_k \\
		&\qquad + (-1)^{(|u|-1)\sum_{a=1}^{k-1}|v_a|}\,
		(-1)^{(|u|+|v_k|-1)\sum_{a=1}^{k-1}|v_a|}\,
		[u,v_k]\wedge v_1\wedge\cdots\wedge v_{k-1} \\[2mm]
		&= \sum_{i=1}^{k} (-1)^{\sum_{\alpha=1}^{i-1}|v_i||v_\alpha|}\,
		[u,v_i]\wedge v_1\wedge\cdots\wedge \hat{v}_i\wedge\cdots\wedge v_k .
	\end{align*}
	The first equality follows from the graded Leibniz rule (Proposition~\ref{s2pro} identity (3)). In the second equality we apply the induction hypothesis to the first term and use graded commutativity of the wedge product to reorder the second. The last equality combines the two sums and simplifies the sign factor. This completes the induction.
\end{proof}

A key operation linking multivector fields to differential forms is the interior product. For a multivector field $v$ and a form $\alpha$, it is defined inductively by
$$\iota_{v_1 \wedge \cdots \wedge v_n} \alpha = \iota_{v_n} \cdots \iota_{v_1} \alpha,$$
and extended by linearity. The degree of $\iota_v$ is $-|v|$. The Lie derivative along a multivector field $v \in \mathfrak{X}^{\wedge\bullet}(M)$ is then defined via the graded Cartan identity:
\begin{equation}
	\mathcal{L}_v\alpha = d\iota_v\alpha - (-1)^{|v|}\iota_v d\alpha.
\end{equation}
A fundamental identity \cite{FORGER_2003} relating the Lie derivative, the interior product, and the Schouten-Nijenhuis bracket is given by
\begin{equation}\label{1e1}
	\iota_{[u,v]}\alpha = (-1)^{(|u|-1)|v|}\mathcal{L}_u\iota_v\alpha - \iota_v\mathcal{L}_u\alpha.
\end{equation}

\subsection{$L_\infty$-algebra}
An $L_\infty$-algebra \cite{KS2022,lada1994,lada1993} generalizes Lie algebras to differential graded vector spaces, replacing the Jacobi identity with a coherent system of higher homotopies.
\begin{definition}
	An $L_\infty$-algebra is a graded vector space $L$ equipped with a family of graded-skew-symmetric linear maps 
	$$\left\{l_k:L^{\otimes k}\to L|1\leq k<\infty\right\}$$
	where $|l_k| = k-2$, satisfying the generalized Jacobi identities for every $m \ge 1$:
	\begin{equation}
		\sum_{\substack{i+j=m+1,\\ \sigma\in Sh\left(i,m-i\right)}}\left(-1\right)^\sigma\epsilon\left(\sigma\right)\left(-1\right)^{i\left(j-1\right)}l_j\left(l_i\left(x_{\sigma\left(1\right)},\cdots,x_{\sigma\left(i\right)}\right),x_{\sigma\left(i+1\right)},\cdots,x_{\sigma\left(m\right)}\right) = 0.
	\end{equation}
\end{definition}

\begin{definition}
	An $L_\infty$-algebra $(L, {l_k})$ is a Lie $n$-algebra if $L$ is concentrated in degrees $0, 1, \dots, n-1$.
\end{definition}

\subsection{Simplicial sets, Kan complexes, and semi-simplicial sets}
\label{subsec:simplicial}

In Section~\ref{s4} we will assign to each $k$-dimensional submanifold a Hamiltonian $k$-form, organized into a simplicial set of observables $\bm{\mathrm{sOb}}_\bullet(M)$. The necessary background on simplicial methods is recalled below; a standard reference is \cite{f2023}.

\begin{definition}
	The simplicial category $\bm{\Delta}$ has objects $[n]=\{0,1,\dots,n\}$ ($n\ge 0$) and order-preserving maps. A simplicial set is a contravariant functor $X_\bullet:\bm{\Delta}^{\mathrm{op}}\to\bm{\mathrm{Set}}$. Thus it consists of sets $X_n$ ($n$-simplices) with face maps $d_i:X_n\to X_{n-1}$ ($0\le i\le n$) and degeneracy maps $s_i:X_n\to X_{n+1}$ ($0\le i\le n$) satisfying the simplicial identities.
\end{definition}

\begin{definition}
	A \textbf{semi-simplicial set} (also called a $\Delta$-set) is a presheaf on the subcategory $\bm{\Delta}_{\mathrm{inj}}$ of $\bm{\Delta}$ containing only injective (i.e., face) maps. Thus it has face maps $d_i$ but no degeneracies.
\end{definition}

A semi-simplicial set can be turned into a simplicial set by freely adding degeneracies; however, if the semi-simplicial set already satisfies the Kan condition (every horn has a filler), then the extension to a simplicial set is unique. More precisely:

\begin{theorem}[Unique degeneracy extension]
	\label{thm:unique-deg}
	Let $X_\bullet$ be a semi-simplicial set that is a Kan complex, i.e., every horn $\Lambda^n_k$ in $X_\bullet$ (defined using only the existing face maps) admits a filler. Then there exists a unique extension to a simplicial set $\tilde X_\bullet$ (i.e., $\tilde X_n = X_n$ and the face maps coincide) such that $\tilde X_\bullet$ is a Kan complex. Uniqueness means that the degeneracy maps are uniquely determined by the requirement that they satisfy the simplicial identities and that degenerate simplices are exactly those that are constant in some direction.
\end{theorem}
This result is classical; the existence was proved in \cite{Rou1971}, and the uniqueness (in the sense above) was established in \cite{McC2013}.

\begin{definition}
	For a topological space $X$, the singular simplicial set $\bm{\mathrm{Sing}}_\bullet(X)$ has $\bm{\mathrm{Sing}}_n(X)=\mathrm{Hom}_{\bm{\mathrm{Top}}}(\Delta^n,X)$. For smooth manifolds, $\bm{\mathrm{Sing}}^\infty_\bullet(X)$ is defined analogously using smooth maps $\Delta^n\to X$. Both are Kan complexes.
\end{definition}

In our construction (Section~\ref{s4}), the observables $\bm{\mathrm{sOb}}_\bullet(M)$ will be defined first as a semi-simplicial set via face maps that restrict Hamiltonian forms to boundaries. Once we prove it satisfies the Kan filling property (Theorem~\ref{s4t1}), Theorem~\ref{thm:unique-deg} guarantees a unique extension to a simplicial set, and the resulting degeneracy maps have a natural geometric interpretation: they insert a direction along which the Hamiltonian form is pulled back from a lower-dimensional submanifold, corresponding to a `constant' observable in that direction.

\subsection{$n$-plectic geometry and observables}\label{s2.4}

A multisymplectic (or $n$-plectic) manifold generalizes a symplectic manifold: while symplectic geometry describes point particles, $n$-plectic geometry governs $(n-1)$-dimensional extended objects~\cite{Baez2010ya,Baez2005qu}. Rogers and others established that Hamiltonian $(n-1)$-forms on an $n$-plectic manifold form an $L_\infty$-algebra~\cite{Rogers_2011,Baez2009}. Here we extend this to Hamiltonian forms of all degrees.

\begin{definition}
	An $(n+1)$-form $\omega$ on a smooth manifold $M$ is $n$-plectic if it is closed and nondegenerate, i.e., $\iota_v\omega = 0 \Rightarrow v = 0$ for every $v \in T_xM$. The pair $(M,\omega)$ is an $n$-plectic manifold.
\end{definition}

The $n$-plectic structure induces bundle maps $T^k: \mathfrak{X}^k(M) \to \Omega^{n-k+1}(M)$, $v \mapsto \iota_v\omega$, leading to the notion of Hamiltonian observables.

\begin{definition}
	Let $(M,\omega)$ be $n$-plectic. An $(n-k)$-form $\alpha \in \Omega^{n-k}(M)$ is \emph{Hamiltonian} if there exists $v_\alpha \in \mathfrak{X}^k(M)$ with $d\alpha = -\iota_{v_\alpha}\omega$. The Poisson brace is $\{\alpha,\beta\} = (-1)^{|v_\beta|}\,\iota_{v_\beta}\iota_{v_\alpha}\omega$. The spaces of such forms and multivector fields are denoted $\Omega^{n-k}_{\mathrm{Ham}}(M)$ and $\mathfrak{X}^k_{\mathrm{Ham}}(M)$.
\end{definition}

Hamiltonian multivector fields satisfy $\mathcal{L}_v\omega = 0$. A Hamiltonian form $\alpha$ determines its associated multivector field $v_\alpha$ only up to elements of $\operatorname{Ker}(T^k)$, while $v_\alpha$ determines $\alpha$ only up to closed forms. This ambiguity reflects a gauge freedom inherent in the choice of observables. Moreover, the Hamiltonian condition $d\alpha = -\iota_v\omega$ may admit solutions only locally on contractible open subsets. We therefore introduce the following notion, which will play a central role in the geometric constructions of Section~\ref{s4}.

\begin{definition}\label{localob}
	A multivector field $v \in \mathfrak{X}^k(M)$ is $n$-plectic if $\mathcal{L}_v\omega = 0$. For such $v$, on any contractible open subset $U \subset M$ there exists a local $(n-k)$-form $\alpha_U \in \Omega^{n-k}(U)$ with $d\alpha_U = -\iota_v\omega|_U$, called a \emph{local observable} associated with $v$. On overlaps, local observables differ by closed forms, and their equivalence classes in local de Rham cohomology encode topological charges~\cite{FD2000,SS2024,SS2023}.
\end{definition}

\paragraph{Working assumptions.}
For the constructions that follow we make two assumptions that guarantee well-definedness.
\begin{enumerate}
	\item \textbf{Topological triviality (Section~\ref{s3}).} The $n$-plectic manifold $(M,\omega)$ is topologically trivial, i.e.\ its de Rham cohomology vanishes in positive degrees. This provides a global primitive $\theta$ with $d\theta = \omega$ and ensures that all Hamiltonian forms are globally defined, simplifying the algebraic description of the $L_\infty$-algebra. The general case is treated by gluing local data.
	\item \textbf{Hamiltonian translations (Sections~\ref{s4}--\ref{s5}).} 
	The infinitesimal generators of worldvolume translations are required to be Hamiltonian vector fields. Physically, the \(n\)-dimensional worldvolume of an extended object carries a natural action of the translation group, and the corresponding conserved Noether charges (energy-momentum) must be represented as genuine observables in the Hamiltonian algebra \(\Omega_{\mathrm{Ham}}^\bullet(M)\). Mathematically, this assumption guarantees that the auxiliary vector fields in the definition of \(\bm{\mathrm{sOb}}_\bullet(M)\) are Hamiltonian; their pairwise commutativity then follows from the fact that they are the images of commuting coordinate vector fields under the tangent map of \(\tau_k\). This commutativity is the essential ingredient in the proof of the Kan property (Theorem~\ref{s4t1}), where the filler is obtained from the top face simply by deleting the generating vector field that corresponds to the inward normal direction. No global topological triviality of \(M\) is required; all constructions rely only on the local Poincar\'e lemma on contractible open sets.
\end{enumerate}
Both assumptions are natural in the context of field theory on extended objects: topological non-triviality can be treated by gluing local data, while the Hamiltonian nature of translations is physically necessary for a consistent momentum observable. We leave the general case for future work.

\section{$L_\infty$-algebra of observable algebras}\label{s3}
In this section we will extend Rogers' work to any degree of Hamiltionian forms and explore the algebraic structure defined by generalized $n$-Poisson brackets. We begin by recalling Rogers' construction of an $L_\infty$-algebra on an $n$-plectic manifold $(M, \omega)$ \cite{Rogers_2011}.
\begin{theorem}[\cite{Rogers_2011}]\label{thm:rogers}
	Given a  $n$-plectic manifold $(M,\omega)$, there exists an $L_\infty$-algebra $(L, {l_k})$ with underlying graded vector space is
	\begin{equation}
		L_i=
		\begin{cases}
			\Omega^{n-1}_{\mathrm{Ham}}\left(M\right)\quad &i=0,\\
			\Omega^{n-i-1}\left(M\right)\quad &0<i\leq n-1,\\
		\end{cases}\nonumber 
	\end{equation}
	and whose structure maps are defined as follows:
	\begin{itemize}
		\item The differential $l_1: L \to L$ is given by
		$$l_1\left(\alpha\right)=d\alpha $$
		if $|\alpha|>0$
		\item All higher maps $\{l_k: L^{\otimes k} \to L \mid 2 \leq k < \infty\}$ are constructed from  the multilinear bracket
		\begin{equation}
			l_k\left(\alpha_1,\cdots,\alpha_k\right)=
			\begin{cases}
				0&|\bigotimes_{i=1}^k\alpha_i|>0;\\
				\left(-1\right)^{\frac{k}{2}+1}\iota_{v_{\alpha_1}\wedge\cdots\wedge v_{\alpha_k}}\omega&|\bigotimes_{i=1}^k\alpha_i|=0, k\,\mathrm{even};\\
				\left(-1\right)^{\frac{k-1}{2}}\iota_{v_{\alpha_1}\wedge\cdots\wedge v_{\alpha_k}}\omega&|\bigotimes_{i=1}^k\alpha_i|=0, k\,\mathrm{odd}.\nonumber\\
			\end{cases}
		\end{equation}
		for $k>1$, where $v_{\alpha_i}$ is the unique Hamiltonian vector field associated to $\alpha_i\in\Omega^{n-1}_{\mathrm{Ham}}\left(M\right)$
	\end{itemize}
\end{theorem}
Theorem \ref{thm:rogers} applies only to Hamiltonian $(n-1)$-forms. To extend this structure to Hamiltonian forms of arbitrary degree, we introduce a degree-shifting mechanism via a formal Grassmann variable with differential degree $-1$, which maps any Hamiltonian $(n-k-1)$-form in $L_k$ to a Hamiltonian $(n-k)$-form in $L_{k-1}$.

\subsection{Extended phase space} As discussed in Section~\ref{s2.4}, the $n$-bracket introduced in Theorem~\ref{thm:rogers} generalizes the classical Poisson bracket of symplectic geometry. In the present section, this bracket is extended to arbitrary Hamiltonian forms. The following lemma establishes that the $n$-bracket of Hamiltonian forms is again a Hamiltonian form. Consequently, the $n$-brackets serve a dual purpose: they define the algebraic structure on the space of Hamiltonian forms, and simultaneously provide the higher homotopy-coherent data of the associated $L_\infty$-algebra. It is therefore essential to introduce a mechanism that shifts all Hamiltonian forms to degree $0$, thereby cleanly separating these two roles: all Hamiltonian forms reside in degree $0$, while the higher degrees encode the higher coherent data.
\begin{lemma}\label{s3l1}
	Given a $n$-plectic manifold $\left(M,\omega\right)$ and $v_1,\cdots,v_m\in\mathfrak{X}^{\bullet}_{\mathrm{Ham}}\left(M\right)$ with $m\geq2$
	\begin{align}
		d\iota_{v_1\wedge\cdots\wedge v_m}\omega=&\sum_{1\leq i<j\leq m}\left(-1\right)^{\sum_{\alpha=1}^{m-j}|v_{j+\alpha}|+\sum_{\alpha=1}^{j-1}\left(|v_j|-1\right)|v_\alpha|+\sum_{\alpha=1}^{i-1}|v_i||v_\alpha|}\iota\left(\left[v_{j},v_i\right]\right.\nonumber\\
		&\qquad\left.\wedge v_1\wedge\cdots\wedge \hat{v}_i\wedge\cdots \wedge \hat{v}_{j}\wedge\cdots\wedge v_{m}\right)\omega
	\end{align}
\end{lemma}
\begin{proof}
	The lemma can be proved by mathematical induction, following an argument similar to that in \cite{Rogers_2011}. For the base case \( m = 2 \),
	\begin{align}\label{s3e1}
		d\iota_{v_1\wedge v_2}\omega&=d\iota_{v_2}\iota_{v_1}\omega\nonumber\\
		&=\mathcal{L}_{v_2}\iota_{v_1}\omega+\left(-1\right)^{\left|v_2\right|}\iota_{v_2}d\iota_{v_1}\omega\nonumber\\
		&=\mathcal{L}_{v_2}\iota_{v_1}\omega\nonumber\\
		&=\left(-1\right)^{\left(\left|v_2\right|-1\right)\left|v_1\right|}\left(\iota_{\left[v_2,v_1\right]}\omega+\iota_{v_1}\mathcal{L}_{v_2}\omega\right)\nonumber\\
		&=\left(-1\right)^{\left(\left|v_2\right|-1\right)\left|v_1\right|}\iota_{\left[v_2,v_1\right]}\omega
	\end{align}
	The first equality is just the definition, the second applies the Cartan identity for the Lie derivative. The third follows from the property of the Hamiltonian vector field $v_2$, The forth equality uses the fundamental identity \eqref{1e1} relating the Lie derivative to the Schouten-Nijenhuis bracket. Finally, the last equality holds because the Hamiltonian multivector field $v_1$ preserves the multisymplectic structure, i.e., $\mathcal{L}_{v_1} \omega = 0$. 
	
	Carrying out the analogous computation for the inductive step gives
	\begin{align}\label{s3e2}
		d\iota_{v_1\wedge\cdots\wedge v_k}\omega
		=\left(-1\right)^{\left(|v_k|-1\right)\sum_{a=1}^{k-1}|v_a|}\iota_{\left[v_k,v_1\wedge\cdots\wedge v_{k-1}\right]}\omega+\left(-1\right)^{|v_k|}\iota_{v_k}d\iota_{v_1\wedge\cdots\wedge v_{k-1}}\omega
	\end{align}
	
	Applying Lemma~\ref{s2l1}, the Schouten-Nijenhuis bracket appearing above expands as
	\begin{align*}
		\left[v_k,v_1\wedge\cdots\wedge v_{k-1}\right]=\sum_{i=1}^{k-1}\left(-1\right)^{\sum_{\alpha=1}^{i-1}|v_i||v_\alpha|}\left[v_k,v_i\right]\wedge v_1\wedge\cdots\hat{v}_i\wedge\cdots \wedge v_{k-1}
	\end{align*}  
	Substituting this expression back, we obtain
	\begin{align}\label{s3e}
		d\iota_{v_1\wedge\cdots\wedge v_k}\omega
		=&\sum_{i=1}^{k-1}\left(-1\right)^{\sum_{\alpha=1}^{k-1}\left(|v_k|-1\right)|v_\alpha|+\sum_{\alpha=1}^{i-1}|v_i||v_\alpha|}\iota\!\left(\left[v_k,v_i\right]\wedge v_1\wedge\!\cdots\!\wedge \hat{v}_i\wedge\!\cdots\! \wedge v_{k-1}\right)\omega\nonumber\\
		&\qquad+\left(-1\right)^{|v_k|}\iota_{v_k}d\iota_{v_1\wedge\cdots\wedge v_{k-1}}\omega
	\end{align}
	By the induction hypothesis,
	\begin{align*}
		d\iota_{v_1\wedge\cdots\wedge v_{k-1}}\omega
		=&\sum_{1\leq i<j\leq k-1}\left(-1\right)^{\sum_{\alpha=1}^{k-j-1}|v_{j+\alpha}|+\sum_{\alpha=1}^{j-1}\left(|v_j|-1\right)|v_\alpha|+\sum_{\alpha=1}^{i-1}|v_i||v_\alpha|}\iota\left(\left[v_{j},v_i\right]\right.\\
		&\qquad\left.\wedge v_1\wedge\cdots\wedge \hat{v}_i\wedge\cdots \wedge \hat{v}_{j}\wedge\cdots\wedge v_{k-1}\right)\omega
	\end{align*}
	Combining the two formulas above yields, for \( m = k \),
	\begin{align*}
		&d\iota_{v_1\wedge\cdots\wedge v_k}\omega\\
		=&\sum_{i=1}^{k-1}\left(-1\right)^{\sum_{\alpha=1}^{k-1}\left(|v_k|-1\right)|v_\alpha|+\sum_{\alpha=1}^{i-1}|v_i||v_\alpha|}\iota\left(\left[v_k,v_i\right]\wedge v_1\wedge\cdots\wedge \hat{v}_i\wedge\cdots \wedge v_{k-1}\right)\omega\\
		&+\left(-1\right)^{|v_k|}\iota_{v_k}\sum_{1\leq i<j\leq k-1}\left(-1\right)^{\sum_{\alpha=1}^{k-j-1}|v_{j+\alpha}|+\sum_{\alpha=1}^{j-1}\left(|v_j|-1\right)|v_\alpha|+\sum_{\alpha=1}^{i-1}|v_i||v_\alpha|}\\
		&\qquad\iota\left(\left[v_{j},v_i\right]\wedge v_1\wedge\cdots\wedge \hat{v}_i\wedge\cdots \wedge \hat{v}_{j}\wedge\cdots\wedge v_{k-1}\right)\omega\\
		=&\sum_{i=1}^{k-1}\left(-1\right)^{\sum_{\alpha=1}^{k-1}\left(|v_k|-1\right)|v_\alpha|+\sum_{\alpha=1}^{i-1}|v_i||v_\alpha|}\iota\left(\left[v_k,v_i\right]\wedge v_1\wedge\cdots\wedge \hat{v}_i\wedge\cdots \wedge \hat{v}_{k}\right)\omega\\
		&+\sum_{1\leq i<j\leq k-1}\left(-1\right)^{\sum_{\alpha=1}^{k-j}|v_{j+\alpha}|+\sum_{\alpha=1}^{j-1}\left(|v_j|-1\right)|v_\alpha|+\sum_{\alpha=1}^{i-1}|v_i||v_\alpha|}\\
		&\iota\left(\left[v_{j},v_i\right]\wedge v_1\wedge\cdots\wedge \hat{v}_i\wedge\cdots \wedge \hat{v}_{j}\wedge\cdots\wedge v_{k}\right)\omega\\
		=&\sum_{1\leq i<j\leq k}\left(-1\right)^{\sum_{\alpha=1}^{k-j}|v_{j+\alpha}|+\sum_{\alpha=1}^{j-1}\left(|v_j|-1\right)|v_\alpha|+\sum_{\alpha=1}^{i-1}|v_i||v_\alpha|}\iota\left(\left[v_{j},v_i\right]\right.\\
		&\qquad\left.\wedge v_1\wedge\cdots\wedge \hat{v}_i\wedge\cdots \wedge \hat{v}_{j}\wedge\cdots\wedge v_{k}\right)\omega.
	\end{align*}
	This completes the inductive step and the proof.
\end{proof}

The degree shift is realized by introducing a formal variable $u$ of bidegree $(-1,1)$ and imposing the condition $du=0$. The first component of this bidegree lowers the form degree, mapping any Hamiltonian form into the degree-zero component 
$L_0$ in Theorem~\ref{thm:rogers}; the second component is a newly added grading that marks the codimension. As a result, the originally graded vector space acquires a bigraded structure, where the extra grading precisely records the codimension information.

Concretely, we extend the original graded space to a bigraded one by assigning to each element of $\Omega^{n-i-1}(M)$ the bidegree $(i,0)$. Denote by $\deg_1$ and $\deg_2$ the first and second degrees of a homogeneous element, and let $|\cdot| = \deg_1 + \deg_2$ be the total degree. The bigraded extension of the vector space $L_{\bullet,\bullet}$ is then defined as
$$
L_{i,j}=
\begin{cases}
	\Omega_{\mathrm{Ham}}^{\,n-j-1}(M)\langle u^{\,j}\rangle, & i = 0,\\[5pt]
	\Omega^{\,n-i-j-1}(M)\langle u^{\,j}\rangle, & i \geq 1,
\end{cases}
$$
where $\Omega(M)\langle u^{\,j}\rangle$ denotes the space of degree-$j$ monomials in $u$ with coefficients in $\Omega(M)$ (equivalently, the free module generated by $u^{\,j}$ over $\Omega(M)$). All Hamiltonian forms are encoded in the subspace $\bigoplus_{j=0}^{n-1} L_{0,j}$; hence the extended phase space is precisely this direct sum, i.e., 

$$\bigoplus_{j=0}^{n-1}L_{0,j}=\bigoplus_{j=0}^{n-1}\Omega_{\mathrm{Ham}}^{\,n-j-1}(M)\langle u^{\,j}\rangle.$$

As discussed in Section~\ref{s1}, an $n$-plectic manifold describes the dynamics of an $(n-1)$-dimensional extended object. In this picture, the exterior derivative of a $k$-form observable on the $n$-plectic geometry can be interpreted as the $k$-plectic structure of a field theory defined on a $k$-dimensional submanifold of the object's world volume. Therefore, the variable $u$ is not merely a formal device---it admits a direct physical interpretation as a fermionic ghost field that generates infinitesimal transversal displacements of the submanifold inside the world volume. Moreover, the exponent of $u$ naturally encodes the codimension of that submanifold relative to the full world volume.

The formal variable $u$ can likewise be used to shift the degree of multivector fields: an element of $\mathfrak{X}^j\left(M\right)$ is moved to total degree $1$ after tensoring with an appropriate power of $u$. We define the extended space of bigraded multivector fields by 
$$\mathfrak{X}^{i,j}\left(M\right):=\mathfrak{X}^{i+j}\left(M\right)\langle u^j\rangle.$$
Let $v_1u^j\in\mathfrak{X}^{i,j}$ and $v_2u^l\in\mathfrak{X}^{i,l}$, the bracket on $\mathfrak{X}^{\bullet,\bullet}(M)$ is extended from bracket on $\mathfrak{X}^{\bullet}(M)$ by following relation 
$$\left[v_1u^j,v_2u^l\right]=\left[v_1,v_2\right]u^{j+l}.$$
For $v_i \in \mathfrak{X}^{i}(M)$,
$$\sum_{i=1}^{n} v_{i}\,u^{\,i-1}\;\in\;\bigoplus_{j=0}^{n-1}\mathfrak{X}^{\,j+1}(M)\langle u^{\,j}\rangle$$ 
The interior product of such an element with the \(n\)-plectic form \(\omega\) is defined by
$$\iota_{\left(\sum_{i=1}^{n} v_{i}u^{i-1}\right)}\omega \;=\; \sum_{i=1}^{n} u^{\,i-1}\,\iota_{v_{i}}\omega .$$
A vector field $v \in \bigoplus_{j=0}^{n-1}\mathfrak{X}^{\,j+1}(M)\langle u^{\,j}\rangle$  is called Hamiltonian if there exists a form $\alpha \in \bigoplus_{j=0}^{n-1} \Omega^{\,n-j-1}(M)\langle u^{\,j}\rangle$ such that
$$\iota_v\omega = -\,d\alpha .$$
The space of Hamiltonian vector fields is denoted by $\bigoplus_{j=0}^{n-1}\mathfrak{X}^{1,j}_{\mathrm{Ham}}(M)$, and the corresponding space of Hamiltonian forms by $\bigoplus_{j=0}^{n-1}\Omega^{n-j-1}_{\mathrm{Ham}}(M)$.

For $k \ge 2$ the $k$-bracket of Hamiltonian forms is defined as
$$l_k(\alpha_1,\dots ,\alpha_k)= (-1)^{\sum_{i=1}^{k}(i-1)(|\alpha_i|+1)}\,\iota_{v_{\alpha_1}\wedge\cdots\wedge v_{\alpha_k}}\omega ,$$
where $v_{\alpha_i}$ is the Hamiltonian vector field associated with $\alpha_i$. For $\alpha_i \in\Omega_{\mathrm{Ham}}^{\,n-i-1}(M)$ with Hamiltonian vector field $v_i \in \mathfrak{X}^{\,i+1}(M)$, the element $\alpha_i u^{\,i}$ belongs to $L_{0,i}$, and the definition of $k$-brackets on $\bigoplus_{j=0}^{n-1} \Omega_{\mathrm{Ham}}^{\,n-j-1}(M)\langle u^{\,j}\rangle$ means
$$l_k\bigl(\alpha_1 u,\dots ,\alpha_i u^{\,i},\dots ,\alpha_k u^{\,k}\bigr)= \Bigl(\prod_{i=1}^{k} u^{\,i}\Bigr)l_k(\alpha_1,\dots ,\alpha_k).$$
Hence, the bracket $l_k$ defined on $\bigoplus_{j=0}^{n-1} \Omega_{\mathrm{Ham}}^{\,n-j-1}(M)\langle u^{\,j}\rangle$ can be understood as the multilinear extension with respect to the formal variable $u$ of the bracket originally defined on $\Omega_{\mathrm{Ham}}^{\bullet}(M)$. In view of the preceding discussion, the variable $u$ acts merely by scalar multiplication (since it commutes with all operations and $du = 0$). Consequently, Lemma~\ref{s3l1} remains valid on the extended space of multivector fields $\mathfrak{X}^{i,j}(M)$. 

Having equipped $L_{\bullet,\bullet}$ with higher brackets, we now recover the original unshifted observables. Recall that the exponent of $u$ encodes the codimension of a submanifold within the world volume. Hence, to extract physical information for a submanifold of codimension $k$, one isolates the degree-$k$ component in $u$. This is achieved by differentiation with respect to $u$: because $u$ is a formal variable, the map
$$\alpha \;\longmapsto\; \left.\frac{d^k\alpha}{du^k}\right|_{u=0}$$
projects an extended observable $\alpha$ in extended phase space $\bigoplus_{j=0}^{n-1}\Omega_{\mathrm{Ham}}^{n-k-1}(M)\langle u^{j}\rangle$ onto its $u^k$ coefficient which resides in the ordinary phase space $\Omega_{\mathrm{Ham}}^{n-k-1}(M)$. Thus, the extended space serves as a generating object for ordinary Hamiltonian forms, and the extraction of coefficients corresponds geometrically to a fiber integration that eliminates the auxiliary direction.

\subsection{Graded $L_\infty$-algebra of observables}\label{s33}

We now extend Theorem~\ref{thm:rogers} to Hamiltonian forms of all degrees. These are encoded in the bigraded space $\bigoplus_{j=0}^{n-1} \Omega_{\mathrm{Ham}}^{\,n-j-1}(M)\langle u^{\,j}\rangle$, where the formal variable $u$ shifts the form degree. Since this space carries a $(\mathbb{Z},\mathbb{Z})$-bigrading $(i,j)$ with total degree $|\cdot| = i+j$, we adopt the following notion.
\begin{definition}
	A \emph{graded $L_\infty$-algebra} is a bigraded vector space $L_{\bullet,\bullet}$ equipped with graded skew-symmetric multilinear maps $l_k : L^{\otimes k} \to L$ ($k \ge 1$) of bidegree $(k-2,0)$ satisfying the generalized Jacobi identities
	\begin{equation}\label{s32e1}
		\sum_{\substack{i+j=m+1,\\ \sigma \in \mathrm{Sh}(i,m-i)}} 
		(-1)^{\sigma}\,\epsilon(\sigma)\,(-1)^{i(j-1)}\,
		l_j\!\bigl(l_i(x_{\sigma(1)},\dots ,x_{\sigma(i)}),\,x_{\sigma(i+1)},\dots ,x_{\sigma(m)}\bigr)=0 .
	\end{equation}
\end{definition}
Our main algebraic result is the following.

\begin{theorem}\label{thmmain}
	Given a $n$-plectic manifold $(M,\omega)$, there exists an graded $L_\infty$-algebra $(L_{\bullet,\bullet}, {l_k})$ with underlying bigraded vector space
	\begin{equation*}
		L_{i,j}=
		\begin{cases}
			\Omega_{\mathrm{Ham}}^{\,n-j-1}(M)\langle u^{\,j}\rangle, & i = 0,\\[10pt]
			\Omega^{\,n-i-j-1}(M)\langle u^{\,j}\rangle, & i \geq 1,
		\end{cases}	
	\end{equation*}
	The structure maps are defined following:
	\begin{itemize}
		\item The differential $l_1: L_{i,j} \to L_{i-1,j}$ is
		\begin{equation*}
			l_1\left(\alpha\right)\!=\!
			\begin{cases}
				d\alpha, &\alpha\in L_{\geq 1,\bullet};\\
				0,&\alpha\in L_{0,\bullet}.
			\end{cases}
		\end{equation*}
		\item The higher maps(or higher Poisson bracket)$\{l_k: L^{\otimes k} \to L \mid 2 \leq k < \infty\}$ are defined from the multilinear bracket  
		\begin{equation*}
			l_k\left(\alpha_1,\cdots,\alpha_k\right)=
			\begin{cases}
				\left(-1\right)^{\sum_{i=1}^k\left(i-1\right)\left(|\alpha_i|+1\right)}\iota_{v_{\alpha_1}\wedge\cdots\wedge v_{\alpha_k}}\omega,&\forall 1\leq i\leq k,\deg_1(\alpha_i)=0\\
				0,&\exists 1\leq i\leq k, \deg_1(\alpha_i)>0
			\end{cases}
		\end{equation*}
		for $k>1$, where $v_{\alpha_i}$ is the Hamiltonian vector field associated to $\alpha_i\in \bigoplus_{j=0}^{n-1} \Omega_{\mathrm{Ham}}^{\,n-j-1}(M)\langle u^{\,j}\rangle$
		where $v_{\alpha_1},\cdots, v_{\alpha_k}$ are the Hamiltonian multivector fields corresponding to $\alpha_1,\cdots, \alpha_k \in L$.
	\end{itemize}
\end{theorem}
\begin{proof} 
	To prove the theorem, we should compute the degree of $l_k$ for any $k$ and then establish its graded skew-symmetry and the homotopy Jacobi identity. 
	
	\noindent\textbf{1. Degree of $l_k$}
	Let $v\in\mathfrak{X}^{1,a}\left(M\right)$ be a Hamiltonian multivector of degree $\left(1,a\right)$, with corresponding Hamiltonian form $\alpha\in\Omega_{\mathrm{Ham}}^{n-a-1}(M)\langle u^a\rangle$ of degree $\left(0,a\right)$. Thus we have the relation
	$$	|v| = |\alpha|+ 1.$$
	It is easy to check that $|l_k|=\left(k-2,0\right)$.
	
	\noindent\textbf{2. Graded skew-symmetry.}
	Since any permutation decomposes into transpositions, it suffices to check the effect of swapping two entries. Let $\alpha_1,\dots,\alpha_k$ be homogeneous Hamiltonian forms. Using $|v_\alpha| = |\alpha|+1$, we compute:
	\begin{align*}
		&l_k\bigl(\alpha_1,\dots,\alpha_i,\cdots,\alpha_j,\dots,\alpha_k\bigr) \\
		=& \; (-1)^{\sum_{a=1}^{k}(a-1)(|\alpha_a|+1)}\;
		\iota_{v_{\alpha_1}\wedge\cdots\wedge v_{\alpha_i}\wedge\cdots\wedge v_{\alpha_j}\wedge\cdots\wedge v_{\alpha_k}}\omega \\[2mm]
		=& \; (-1)^{\sum_{a=1}^{k}(a-1)(|\alpha_a|+1)}\;
		(-1)^{|v_{\alpha_i}||v_{\alpha_j}|+\sum_{b=i+1}^{j-1}|v_{\alpha_b}|(|v_{\alpha_i}|+|v_{\alpha_j}|)} \\
		& \qquad\qquad \iota_{v_{\alpha_1}\wedge\cdots\wedge v_{\alpha_j}\wedge,\cdots,\wedge v_{\alpha_i}\wedge\cdots\wedge v_{\alpha_k}}\omega \\[2mm]
		=& (-1)^{\left(j-i\right)\left(|v_{\alpha_i}|-|v_{\alpha_j}|\right)+|v_{\alpha_i}||v_{\alpha_j}|+\sum_{b=i+1}^{j-1}|v_{\alpha_b}|(|v_{\alpha_i}|+|v_{\alpha_j}|)}
		l_k\!\bigl(\alpha_1,\!\dots\!,\!\alpha_j,\!\cdots\!,\alpha_i,\!\dots\!,\!\alpha_k\bigr) \\[2mm]
		=& \; -\,(-1)^{|\alpha_i||\alpha_j|+\sum_{b=i+1}^{j-1}|\alpha_b|(|\alpha_i|+|\alpha_j|)}\;
		l_k\bigl(\alpha_1,\dots,\alpha_j,\cdots,\alpha_i,\dots,\alpha_k\bigr) \\[2mm]
		=& \; -\epsilon(\alpha_i\leftrightarrow\alpha_j)\;
		l_k\bigl(\alpha_1,\dots,\alpha_j,\cdots,\alpha_i,\dots,\alpha_k\bigr),
	\end{align*}
	where $\epsilon(\alpha_i\leftrightarrow\alpha_j)$ is the Koszul sign associated with the transposition. This is precisely the required graded skew-symmetry.
	
	\noindent\textbf{3. Homotopy Jacobi identity \eqref{s32e1}.} For $m=1$ the identity reduces to $l_1^2=0$, i.e. $d^2=0$, which holds trivially.
	
	For $m=2$ the identity reads
	$$	l_1\bigl(l_2(\alpha_1,\alpha_2)\bigr)
	- l_2\bigl(l_1(\alpha_1),\alpha_2\bigr)
	+ (-1)^{|\alpha_1||\alpha_2|}l_2\bigl(l_1(\alpha_2),\alpha_1\bigr)=0.$$
	If $\alpha_1,\alpha_2\in L_{0,\bullet}$, $l_2(\alpha_1,\alpha_2)\in L_{0,\bullet}$, hence each term vanishes individually. If $\alpha_1$ or $\alpha_2\in L_{\geq 2,\bullet}$, each term also vanishes.
	
	For $m>2$ we use the fact that the Hamiltonian multivector field associated with $l_1(v_\alpha)$ is zero. This simplifies the general identity \eqref{s32e1} to
	\begin{align}\label{s3e3}
		&dl_m(\alpha_1,\dots,\alpha_m)\nonumber\\
		=& -\sum_{\sigma\in\mathrm{Sh}(2,m-2)}
		(-1)^{\sigma}\epsilon(\sigma)l_{m-1}\!\Bigl(l_2\bigl(\alpha_{\sigma(1)},\alpha_{\sigma(2)}\bigr),\alpha_{\sigma(3)},\dots,\alpha_{\sigma(m)}\Bigr).
	\end{align}
	By Lemma~\ref{s3l1}
	\begin{align*}
		dl_2\left(\alpha_{\sigma(1)},\alpha_{\sigma(2)}\right)=&\left(-1\right)^{|\alpha_{\sigma(2)}|+1}d\iota_{\alpha_{\sigma(1)}\wedge\alpha_{\sigma(2)}}\omega\\
		=&\left(-1\right)^{|\alpha_{\sigma(2)}|+1}\left(-1\right)^{\left(|v_{\alpha_{\sigma(2)}}|-1\right)|v_{\alpha_{\sigma(1)}}|}\iota_{\left[v_{\alpha_{\sigma(2)}},v_{\alpha_{\sigma(1)}}\right]}\omega
	\end{align*}
	Hence,
	\begin{align*}
		v_{l_2(\alpha_{\sigma(1)},\alpha_{\sigma(2)})}=\left(-1\right)^{|\alpha_{\sigma(1)}||\alpha_{\sigma(2)}|}\left[v_{\alpha_{\sigma(2)}},v_{\alpha_{\sigma(1)}}\right]
	\end{align*}

	The right-hand side of \eqref{s3e3} expands as
	\begin{align*}
		-\sum_{\sigma\in \mathrm{Sh}(2,m-2)}&
		(-1)^{\sigma}\epsilon(\sigma)(-1)^{\sum_{a=3}^{m-1}(|\alpha_{\sigma(a)}|+1)(a-2)+|\alpha_{\sigma(1)}||\alpha_{\sigma(2)}|} \iota\Big([v_{\alpha_{\sigma(2)}},v_{\alpha_{\sigma(1)}}]\\
		&\wedge v_{\alpha_{\sigma(3)}}\wedge\cdots\wedge v_{\alpha_{\sigma(m)}}\Big)\omega .
	\end{align*}
	In a shuffle permutation $\sigma$ with $\sigma(1)=i,\; \sigma(2)=j\;(i<j)$,
	$$(-1)^{\sigma}=(-1)^{i+j-1}, \qquad
	\epsilon(\sigma)=(-1)^{|\alpha_i|\sum_{a=1}^{i-1}|\alpha_a|
		+|\alpha_j|\bigl(\sum_{b=1}^{j-1}|\alpha_b|-|\alpha_i|\bigr)}.$$
	Substituting these expressions and carefully collecting all sign factors, one finds that the right-hand side of \eqref{s3e3} equals
	$$	\sum_{1\le i<j\le m}
	(-1)^{\eta(i,j,\{\alpha\})}\,
	\iota\!\Bigl([v_{\alpha_j},v_{\alpha_i}]\wedge
	v_{\alpha_1}\wedge\cdots\wedge\hat{v}_{\alpha_i}
	\wedge\cdots\wedge\hat{v}_{\alpha_j}\wedge\cdots\wedge v_{\alpha_m}\Bigr)\omega ,$$
	where $\eta(i,j,\{\alpha\})$ denotes the total sign exponent arising from the permutation and the grading.
	\begin{align*}
		\eta=&i+j+|\alpha_i|\sum_{a=1}^{i-1}|\alpha_a|+|\alpha_j|\left(\sum_{b=1}^{j-1}|\alpha_a|-|\alpha_i|\right)+\sum_{a=1}^{i-1}a\left(|\alpha_a|+1\right)\\
		&+\sum_{a=i+1}^{j-1}\left(a-1\right)\left(|\alpha_a|+1\right)+\sum_{a=j+1}^{m}\left(a-2\right)\left(|\alpha_a|+1\right)+|\alpha_i||\alpha_j|\\
		=&\sum_{a=1}^{i-1}\left(|\alpha_i|+1\right)\left(|\alpha_a|+1\right)+\sum_{a=1}^{j-1}|\alpha_j|\left(|\alpha_a|+1\right)+\sum_{a=1}^{m}\left(a-1\right)\left(|\alpha_a|+1\right)\\
		&+\sum_{a=j+1}^{m}\left(|\alpha_a|+1\right)-2\sum_{a=j+1}^{m}\left(|\alpha_a|+1\right)-2\left(i-1\right)|\alpha_i|-2\left(j-1\right)|\alpha_j|+2
	\end{align*}
	
	On the other hand, directly computing the left-hand side of \eqref{s3e3} using Lemma~\ref{s3l1} :
	\begin{align*}
		&dl_m\left(\alpha_1,\cdots,\alpha_m\right)\\
		=&\sum_{1\leq i<j\leq m}\left(-1\right)^{\sum_{a=1}^m\left(|\alpha_a|+1\right)\left(a-1\right)+\sum_{a=1}^{m-j}|v_{j+a}|+\sum_{a=1}^{j-1}\left(|v_j|-1\right)|v_a|+\sum_{a=0}^{i-1}|v_i||v_a|}\\
		&\iota\left(\left[v_{\alpha_{j}},v_{\alpha_i}\right]\wedge v_{\alpha_1}\wedge\cdots\wedge \hat{v}_{\alpha_i}\wedge\cdots \wedge \hat{v}_{\alpha_j}\wedge\cdots\wedge v_{\alpha_m}\right)\omega\\
		=&\!\!\!\!\!\!\!\sum_{1\leq i<j\leq m}\!\!\!\!\!\!\!\left(-1\right)^{\sum_{a=1}^m\left(|\alpha_a|+1\right)\left(a-1\right)+\sum_{a=1}^{m-j}\left(|\alpha_{j+a}|+1\right)+\sum_{a=1}^{j-1}|\alpha_j|\left(|\alpha_a|+1\right)+\sum_{a=0}^{i-1}\left(|\alpha_i|+1\right)\left(|\alpha_a|+1\right)}\\
		&\iota\left(\left[v_{\alpha_{j}},v_{\alpha_i}\right]\wedge v_{\alpha_1}\wedge\cdots\wedge \hat{v}_{\alpha_i}\wedge\cdots \wedge \hat{v}_{\alpha_j}\wedge\cdots\wedge v_{\alpha_m}\right)\omega
	\end{align*}
	Hence \eqref{s3e3} holds, and therefore the full homotopy Jacobi identity \eqref{s32e1} is satisfied for all $m\ge 1$.
\end{proof}

\begin{remark}[Higher Heisenberg algebra]\label{rem:heisenberg}
	Let $\tilde{\mathfrak{X}}^{1}_{\mathrm{Ham}}(M)$ be a subspace of Hamiltonian vector fields in which every pair commutes, and let $\tilde{\mathfrak{X}}^{\bullet}_{\mathrm{Ham}}(M)$ be the graded-commutative ring freely generated by these vector fields under addition and wedge product. Denote by $\tilde{\Omega}^{\bullet}_{\mathrm{Ham}}(M)$ the corresponding Hamiltonian forms. Because all Schouten--Nijenhuis brackets among generators vanish, Lemma~\ref{s3l1} implies $d l_n=0$ for every $n\ge 2$; moreover, any nested bracket vanishes since the associated Hamiltonian multivector field of the inner bracket is zero. Consequently, $(\tilde{\Omega}^{\bullet}_{\mathrm{Ham}}(M),\{l_k\})$ forms an $L_\infty$-algebra in which each $l_k$ is closed and all higher Jacobi identities hold trivially. This structure generalizes the classical Heisenberg algebra: for $n=1$ (symplectic geometry) the only non-trivial bracket is $l_2$, which coincides with the constant Poisson bracket $1$.
	
	As will be shown in Section~\ref{s411}, this higher Heisenberg algebra naturally underlies the definition of the semi-simplicial set $\mathbf{sOb}_\bullet(M)$.
\end{remark}

\section{Geometric constructions of observable algebras}\label{s4}

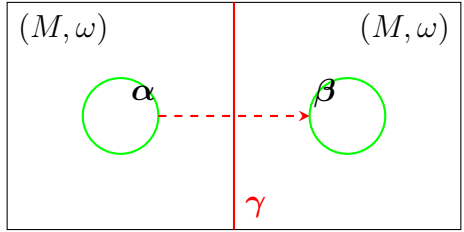
\begin{figure}[htbp]
	\centering
	\begin{tikzpicture}
		\draw (0,0) rectangle (6,3);
		
		\draw[red, thick] (3,0) -- (3,3) node[pos=0.1, right] {$\bm{\gamma}$};
		
		\draw[green, thick] (1.5,1.5) circle (0.5);
		\node at (1.8,1.8) {$\bm{\alpha}$};
		
		\draw[green, thick] (4.5,1.5) circle (0.5);
		\node at (4.2,1.8) {$\bm{\beta}$};
		
		\draw[red, dashed, ->, thick, >=stealth] (2.0,1.5) -- (4.0,1.5);
		
		\node[below right] at (0,3) {$(M,\omega)$};
		\node[below left] at (6,3) {$(M,\omega)$};
	\end{tikzpicture}
	\caption{A defect $\gamma$ in the worldvolume. When an observable $\alpha$ crosses the defect, it transforms into $\beta = l_2(\alpha,\gamma)$.}\label{fig1}
\end{figure}
\subsection{Observables as defects}\label{s40}
Similar to the framework of generalized global symmetries \cite{G2014} and higher charges \cite{B2023,B2024,B2025}, where symmetries and charges are realized as topological defects and their action corresponds to a charge crossing a defect, we adopt an analogous picture. In our setting, $k$-form Hamiltonian observables are interpreted as $k$-dimensional defects. Under orientation reversal, the Hamiltonian associated with a manifold changes sign --- a property that will be relevant when considering defect orientation. When an observable $\alpha$ crosses a defect $\gamma$, it transforms into $\beta := l_2(\alpha,\gamma)$, encoding the defect's action (Figure~\ref{fig1}). This suggests that the algebraic structure of Hamiltonian forms encodes the properties of a global symmetry; indeed, the conserved charges of such a symmetry are precisely given by Hamiltonian forms. A natural consistency condition follows: if $\alpha$ can be continuously deformed around $\gamma$ without crossing, then $l_2(\alpha,\gamma)=0$, reflecting trivial action. Higher algebraic relations, such as the homotopy Jacobi identity, admit a geometric interpretation: they reflect the fact that the order in which an observable crosses a network of defects can be rearranged --- a phenomenon reminiscent of the Yang-Baxter equation and braid group relations.

\begin{figure}[htbp]
	\centering
	\begin{tikzpicture}
		\fill[green!20] (0,0) -- (1.75,0.25) -- (1.75,3.25) -- (0,3) -- cycle;
		\fill[blue!20] (1.75,0.25) -- (3.5,0.5) -- (3.5,3.5) -- (1.75,3.25) -- cycle;
		\draw[thick] (0,0) -- (3.5,0.5) -- (3.5,3.5) -- (0,3) -- cycle;
		\draw[dashed, thick] (1.75,0.25) -- (1.75,3.25) node[pos=0.9,right] {$\gamma$};
		\node at (0.8,1.5) {A, $\alpha$};
		\node at (2.5,1.8) {B,$\beta$};
		\begin{scope}[xshift=8cm]
			\draw[ thick] (0,0) -- (0,3) node[pos=0.5,left] {$\gamma$ };
			
			\fill[blue!20] (0,0) -- (2,0.3) -- (2,3.3) -- (0,3) -- cycle;
			\draw[thick] (0,0) -- (2,0.3) -- (2,3.3) -- (0,3) -- cycle;
			\node at (1,1.7) {B};
			
			\fill[green!40, opacity=0.7] (0,0) -- (2,-0.4) -- (2,2.6) -- (0,3) -- cycle;
			\draw[thick] (0,0) -- (2,-0.4) -- (2,2.6) -- (0,3) -- cycle;
			\node at (1,1.0) {A};
			\node[below] at (1,-0.5) {$\beta-\alpha$};
		\end{scope}
		
		\draw[->, thick, >=Stealth] (5.5,1.5) -- (6.5,1.5) node[midway, above] {fold};
	\end{tikzpicture}
	\caption{Folding construction for gluing two $(n-1)$-dimensional submanifolds of $(M,\omega)$. Left: two submanifolds $(A,\alpha)$ and $(B,\beta)$ separated by an $(n-2)$-dimensional interface $\gamma$. Right: after folding, the two submanifolds are superimposed, and the data on the interface is given by the difference $\beta - \alpha$. }\label{fig2}
\end{figure}
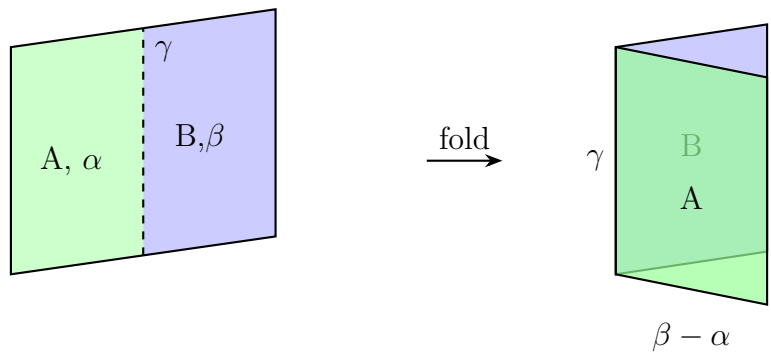

We now recursively assign a collection of Hamiltonian forms to submanifolds of the worldvolume. Starting from the top, the $n$-form $\theta$, satisfying $d\theta = \omega$, serves as an observable on the full worldvolume. On codimension-$1$ regions $A$ and $B$ of the worldvolume, we place $(n-1)$-form Hamiltonian observables $\alpha$ and $\beta$, respectively. Suppose these two regions are joined along a common codimension-$1$ interface, as illustrated in Figure~\ref{fig2}. The observable on this interface must be an $(n-2)$-form $\gamma$ for which there exists a vector field $v_\gamma$ such that
\[
\iota_{v_\gamma}(d\beta-d\alpha) = -d\gamma.
\]
This condition encodes the compatibility of the Hamiltonian data across the interface. Since $\alpha$ and $\beta$ are Hamiltonian forms with associated Hamiltonian multivector fields $v_\alpha$ and $v_\beta$, we have $\iota_{v_\alpha}\omega = -d\alpha$ and $\iota_{v_\beta}\omega = -d\beta$. Substituting these into the compatibility condition yields
\[
\iota_{(v_\alpha - v_\beta) \wedge v_\gamma}\omega = -d\gamma,
\]
which shows that $\gamma$ is itself a Hamiltonian form on $(M,\omega)$.

When more than two regions meet at a common interface, the construction must account for higher-order junctions. Consider $n$ regions $A_1,\dots,A_n$ with Hamiltonian forms $\theta_1,\dots,\theta_n$, respectively, meeting pairwise along codimension-$1$ interfaces labeled by Hamiltonian forms $\alpha_{12},\dots,\alpha_{1n}$, which are themselves constructed recursively as above, satisfying
\[
\iota_{v_{\alpha_{i(i+1)}}}(d\theta_{i+1} - d\theta_i) = -d\alpha_{i(i+1)}.
\]
These interfaces in turn meet at a common codimension-$2$ junction $\delta$, as illustrated in Figure~\ref{fig3}. By rotating all interfaces so that they overlap, we obtain a compatibility condition for $\delta$:
\[
\iota_{v_\delta}\Bigl(\sum_i \pm d\alpha_{i(i+1)}\Bigr) = -d\delta.
\]
The $\pm$ sign is determined by the orientation of each interface relative to the overlapping surface obtained by this rotation: if the orientation coincides, the sign is $+$; if opposite, it is $-$. As in the codimension-$1$ case, one verifies that $\delta$ is also a Hamiltonian form. Proceeding inductively, this construction associates a Hamiltonian form to every submanifold of the worldvolume, regardless of its codimension.

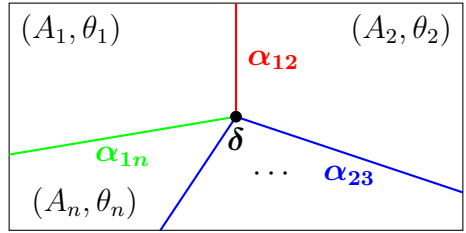
\begin{figure}[htbp]
	\centering
	\begin{tikzpicture}
		\draw (0,0) rectangle (6,3);
		\draw[red, thick] (3,1.5) -- (3,3) node[pos=0.5, right] {$\bm{\alpha_{12}}$};
		\draw[green,thick] (0,1)--(3,1.5) node[pos=0.5,below] {$\bm{\alpha_{1n}}$};
		\draw[blue,thick] (3,1.5)--(6,0.5) node[pos=0.5,below] {$\bm{\alpha_{23}}$};
		\draw[blue,thick] (3,1.5)--(2,0);
		\filldraw[black] (3,1.5) circle (2pt) node[below] {$\bm{\delta}$};
		\node[below right] at (0,3) {$(A_1,\theta_1)$};
		\node[below left] at (6,3) {$(A_2,\theta_2)$};
		\node[above] at (1,0) {$(A_n,\theta_n)$};
		\node[above] at (3.5,0.5) {$\cdots$};
	\end{tikzpicture}
	\caption{Multiple regions $A_1,\dots,A_n$ with Hamiltonian forms $\theta_i$ joined along common interfaces labeled by $\alpha_{ij}$, meeting at a junction $\delta$.}\label{fig3}
\end{figure}

\subsection{Semi-simplicial set of observables}\label{s41}
The recursive procedure described in Section~\ref{s40} for assigning Hamiltonian forms to submanifolds can be naturally extended to geometric simplices, thereby reflecting the structure of a simplicial set of observables.

\subsubsection{Construction of the semi-simplicial set}\label{s411}

Recall from Section~\ref{s2.4} that we assume the infinitesimal generators of worldvolume translations are Hamiltonian vector fields. Physically, the conserved Noether charges (energy-momentum) must be represented as genuine observables. Mathematically, this assumption guarantees that the auxiliary vector fields \(v_1,\dots,v_{n-k}\) appearing in the definition of \(\bm{\mathrm{sOb}}_\bullet(M)\) are Hamiltonian. Moreover, these vector fields are obtained as the images of the coordinate vector fields on \(\Delta^k\times[0,1]^{n-k}\) under the tangent map of \(\tau_k\); because the coordinate vector fields pairwise commute and the tangent map preserves Lie brackets, the resulting vector fields are pairwise commuting Hamiltonian vector fields on the image of \(\tau_k\). These commuting vector fields are precisely the generators of the higher Heisenberg algebra introduced in Remark~\ref{rem:heisenberg}. Their commutativity, the defining feature of that algebra, is the essential ingredient in the proof of the Kan property (Theorem~\ref{s4t1}), where the filler is defined by omitting the generating vector field corresponding to the inward normal direction. No global topological triviality of \(M\) is required; all constructions rely only on the local Poincar\'e lemma on contractible open sets.

Throughout this section we work with these commuting Hamiltonian vector fields. Their commutativity underlies all the simplifications that follow. In particular, the recursive compatibility conditions that determine the Hamiltonian forms on interfaces, junctions, and higher-codimension strata reduce to a uniform rule: the Hamiltonian form on a junction is obtained by contracting the differentials of the Hamiltonian forms on the intersecting branches with the transverse vector fields that generate the motion along those branches. For instance, in the setting of Figure~\ref{fig2}, the vector fields \(v_\alpha\) and \(v_\beta\) are perpendicular to the interface, lie respectively in regions \(A\) and \(B\), and commute. With \(d\alpha = -\iota_{v_\beta}\omega\) and \(d\beta = -\iota_{v_\alpha}\omega\), the Hamiltonian form \(\gamma\) on the interface satisfies \(\iota_{v_\alpha}d\alpha = \iota_{v_\beta}d\beta = -d\gamma\), and analogous simplifications hold for higher-codimension junctions.

Based on this discussion, the simplicial set of observables can now be constructed as follows.

\begin{definition}[Semi-simplicial set of observables]\label{s4d1}
	Let \((M,\omega)\) be an \(n\)-plectic manifold. The semi-simplicial set of observables on \((M,\omega)\), denoted \(\bm{\mathrm{sOb}}_\bullet(M)\), is defined as follows. For each \(0\leq k\leq n\), the set of \(k\)-simplices is
	\[
	\bm{\mathrm{sOb}}_k(M) \;:=\left\{ \sigma_k^*\alpha \;\middle|\;
	\begin{aligned}
		&\sigma_k:\Delta^k \;\hookrightarrow\; \Delta^k\times[0,1]^{n-k} \;\xrightarrow{\;\tau_k\;}\; M \text{ smooth},\\
		&\alpha\in\Omega_{\mathrm{Ham}}^k(U),\; U \text{ an open neighbourhood of } \tau_k(\Delta^k\times[0,1]^{n-k}),\\
		&d\alpha = -\iota_{v_1\wedge\cdots\wedge v_{n-k}}\,\omega|_U,
	\end{aligned}\right\},
	\]
	where \(\tau_k\) maps each \([0,1]\)-factor to a flow line of a Hamiltonian vector field \(v_j\) on \(M\) (\(j=1,\dots,n-k\)), and \(v_1,\dots,v_{n-k}\) are the associated pushforward vector fields along the \([0,1]^{n-k}\) directions. By the Hamiltonian translation hypothesis (Section~\ref{s2.4}), these vector fields are Hamiltonian and pairwise commuting. The pullback \(\sigma_k^*\alpha\) is taken along \(\sigma_k\); the \([0,1]^{n-k}\) directions are auxiliary and serve solely to encode the Hamiltonian translation data that guarantee the existence of the face maps. By convention \(\Omega_{\mathrm{Ham}}^n(U)=\{\theta|_U\}\).
	
	The face maps \(d^i:\bm{\mathrm{sOb}}_k(M)\to\bm{\mathrm{sOb}}_{k-1}(M)\) for \(i=0,\dots,k\) are defined as follows.
	Let \(\partial_i:\Delta^{k-1}\hookrightarrow\Delta^k\) be the inclusion of the \(i\)-th face, and let \(v\) be the inward unit normal vector field along this face with pushforward \(\tilde v = d\tau_k(v)\). The underlying singular simplex of the face is \(\sigma_{k-1} = \tau_{k-1}\circ\iota_{k-1}\), where \(\iota_{k-1}: \Delta^{k-1}\hookrightarrow\Delta^{k-1}\times[0,1]^{n-k+1}\) is the natural embedding and \(\tau_{k-1}\) is obtained from \(\tau_k\) via
	\[
	\tau_{k-1}(p,t,x_1,\dots,x_{n-k}) = \psi_{\tilde v}^t\!\bigl(\tau_k(\partial_i(p),\,x_1,\dots,x_{n-k})\bigr),
	\]
	with \(\psi_{\tilde v}^t\) the flow of \(\tilde v\). Thus \(\sigma_{k-1} = \sigma_k\circ\partial_i\) as maps on \(\Delta^{k-1}\), while the auxiliary directions increase by one to accommodate the normal direction.
	
	Because the Hamiltonian vector fields \(v_1,\dots,v_{n-k},\tilde v\) are pairwise commuting, the multivector field \(v_1\wedge\cdots\wedge v_{n-k}\wedge\tilde v\) is Hamiltonian on \(U\). Consequently,
	\[
	\iota_{\tilde v}\,d\alpha = \iota_{v_1\wedge\cdots\wedge v_{n-k}\wedge\tilde v}\,\omega
	\]
	is closed. Hence there exists a \((k-1)\)-form \(\eta\) on \(U\), unique up to closed forms on \(U\), satisfying
	\[
	d\eta = -(-1)^i\,\iota_{\tilde v}\,d\alpha.
	\]
	The face map is defined by
	\[
	d^i(\sigma_k^*\alpha) := \sigma_{k-1}^*\,\eta.
	\]
	Note that \(\eta\) satisfies the Hamiltonian condition with respect to the vector fields \(\{v_1,\dots,v_{n-k},\tilde v\}\) along the auxiliary \([0,1]^{n-k+1}\) directions, so the image indeed lies in \(\bm{\mathrm{sOb}}_{k-1}(M)\).
	
	For \(k>n\) we make no definition. Thus \(\bm{\mathrm{sOb}}_\bullet(M)\) is an \(n\)-truncated semi-simplicial set.
\end{definition}

\begin{remark}[Well-definedness and role of the auxiliary directions]
	The auxiliary \([0,1]^{n-k}\) directions in the definition of \(\tau_k\) encode the Hamiltonian translation data: the condition \(d\alpha = -\iota_{v_1\wedge\cdots\wedge v_{n-k}}\omega\) selects precisely those Hamiltonian forms whose generating multivector field is the wedge product of the commuting translations along these directions. Together with the inward normal vector field on \(\Delta^k\), these translations form a commuting family of Hamiltonian vector fields, which guarantees that \(\iota_{\tilde v}\,d\alpha\) is exact. Therefore the required \((k-1)\)-form \(\eta\) exists and is unique modulo closed forms, making the face maps well-defined.
\end{remark}

\begin{theorem}
	$\bm{\mathrm{sOb}}_\bullet(M)$ is an $n$-truncated semi-simplicial set (i.e., the face maps satisfy $d^i d^j = d^{j-1} d^i$ for $i<j$).
\end{theorem}

\begin{proof}
	Let $x = \sigma_k^*\alpha \in \bm{\mathrm{sOb}}_k(M)$ and $i<j$. Denote by $v_j$ the unit vector field on $\Delta^k$ normal to the $j$-th face $\partial_j\Delta^k$ and pointing inward, and by $v_i$ the unit vector field on $\Delta^{k-1}$ normal to the $i$-th face $\partial_i\Delta^{k-1}$ and pointing inward. By definition of the face map,
	\[ 
	d^j(x) = \sigma_{k-1}^*\eta_j,\qquad d\eta_j = -(-1)^j\,\iota_{d\sigma_k(v_j)}\,d\alpha,
	\]
	where $\sigma_{k-1} = \sigma_k\circ\partial_j$. Applying $d^i$ gives
	\[
	d^i d^j(x) = \sigma_{k-2}^*\eta_{ij},\qquad d\eta_{ij} = -(-1)^i\,\iota_{d\sigma_{k-1}(v_i)}\,d\eta_j,
	\]
	with $\sigma_{k-2} = \sigma_{k-1}\circ\partial_i = \sigma_k\circ\partial_j\circ\partial_i$.
	
	On the other hand, applying first $d^i$ and then $d^{j-1}$ yields
	\[
	d^i(x) = \sigma_{k-1}'^*\eta_i,\qquad d\eta_i = -(-1)^i\,\iota_{d\sigma_k(v_i)}\,d\alpha,
	\]
	\[
	d^{j-1} d^i(x) = \sigma_{k-2}^*\eta'_{ij},\qquad d\eta'_{ij} = -(-1)^{j-1}\,\iota_{d\sigma_{k-1}'(v_{j-1})}\,d\eta_i,
	\]
	where $\sigma_{k-1}' = \sigma_k\circ\partial_i$ and $v_{j-1}$ is the unit vector field on $\Delta^{k-1}$ normal to the $(j-1)$-st face.
	
	Using $\sigma_{k-1} = \sigma_k\circ\partial_j$, we have $d\sigma_{k-1}(v_i) = d\sigma_k\bigl(d\partial_j(v_i)\bigr)$. Set $(\partial_j)_* v_i := d\partial_j(v_i)$, which is a vector field on $\partial_j\Delta^k$. Then
	\[
	\iota_{d\sigma_{k-1}(v_i)}\,d\eta_j = -(-1)^j\,\iota_{d\sigma_k((\partial_j)_* v_i)}\iota_{d\sigma_k(v_j)}d\alpha 
	= -(-1)^j\,\iota_{d\sigma_k\bigl((\partial_j)_* v_i \wedge v_j\bigr)}d\alpha.
	\]
	Similarly,
	\[
	\iota_{d\sigma_{k-1}'(v_{j-1})}\,d\eta_i = -(-1)^i\,\iota_{d\sigma_k((\partial_i)_* v_{j-1})}\iota_{d\sigma_k(v_i)}d\alpha 
	= -(-1)^i\,\iota_{d\sigma_k\bigl((\partial_i)_* v_{j-1} \wedge v_i\bigr)}d\alpha.
	\]
	
	On $\Delta^k$, the commutativity of face inclusions $\partial_j\circ\partial_i = \partial_i\circ\partial_{j-1}$ (for $i<j$) implies that the composite maps give the same $(k-2)$-face $F = \partial_i(\partial_{j-1}\Delta^{k-1}) = \partial_j(\partial_i\Delta^{k-1})$. The bivector fields $(\partial_j)_* v_i \wedge v_j$ and $(\partial_i)_* v_{j-1} \wedge v_i$ both represent, in a neighborhood of $F$, the normal bivector to $F$ pointing into $\Delta^k$, and by antisymmetry of the wedge product,
	\[
	(\partial_j)_* v_i \wedge v_j = -(\partial_i)_* v_{j-1} \wedge v_i.
	\]
	Hence,
	\[
	\iota_{d\sigma_{k-1}(v_i)}\,d\eta_j = -(-1)^{j-i}\,\iota_{d\sigma_{k-1}'(v_{j-1})}\,d\eta_i.
	\]
	
	Now compute $d\eta_{ij}$ and $d\eta'_{ij}$:
	\[
	d\eta_{ij} = -(-1)^i\,\iota_{d\sigma_{k-1}(v_i)}\,d\eta_j 
	= (-1)^i\cdot(-1)^{j-i}\,\iota_{d\sigma_{k-1}'(v_{j-1})}\,d\eta_i 
	= (-1)^j\,\iota_{d\sigma_{k-1}'(v_{j-1})}\,d\eta_i,
	\]
	\[
	d\eta'_{ij} = -(-1)^{j-1}\,\iota_{d\sigma_{k-1}'(v_{j-1})}\,d\eta_i 
	= -(-1)^j\cdot(-1)^{-1}\,\iota_{d\sigma_{k-1}'(v_{j-1})}\,d\eta_i 
	= (-1)^j\,\iota_{d\sigma_{k-1}'(v_{j-1})}\,d\eta_i.
	\]
	Thus $d\eta_{ij} = d\eta'_{ij}$, which implies $d(\eta_{ij} - \eta'_{ij}) = 0$, i.e., $\eta_{ij} - \eta'_{ij}$ is a closed form. By definition of $\bm{\mathrm{sOb}}_{k-2}(M)$ (elements are equivalence classes of Hamiltonian forms modulo closed forms), $\eta_{ij}$ and $\eta'_{ij}$ represent the same equivalence class, and therefore
	\[
	d^i d^j(x) = d^{j-1} d^i(x) \quad \text{in } \bm{\mathrm{sOb}}_{k-2}(M).
	\]
	This verifies the face-face identities, so $\bm{\mathrm{sOb}}_\bullet(M)$ is a semi-simplicial set.
\end{proof}

\begin{theorem}\label{s4t1}
	\(\bm{\mathrm{sOb}}_\bullet(M)\) is a semi-Kan complex (i.e., a semi-simplicial set satisfying the Kan filling property).
\end{theorem}

\begin{proof}
	Let \(f:\Lambda^m_r \to \bm{\mathrm{sOb}}_\bullet(M)\) be a horn with missing \(r\)-th face,
	where \(0\le r\le m\) and \(1\le m\le n\).
	
	\medskip
	\noindent\textbf{Reduction to the case \(r=0\).}
	Let \(\rho:[m]\to[m]\) be the unique order-preserving bijection that sends \(r\) to \(0\) and
	preserves the cyclic order of the remaining vertices.  Concretely, \(\rho\) is the composition
	of the cyclic permutation \((0\;1\;\cdots\;r)\) with the standard identification of
	\([m]\setminus\{r\}\) and \([m-1]\).  The map \(\rho\) induces an isomorphism of the horn
	\(\Lambda^m_r\) with the horn \(\Lambda^m_0\) via precomposition with the induced map on
	simplices.  Because \(\bm{\mathrm{sOb}}_\bullet(M)\) is a semi-simplicial set, this
	isomorphism is compatible with the face maps.  Hence we may assume, without loss of
	generality, that the missing face is \(\partial_0\Delta^m\).
	
	\medskip
	\noindent Now let \(f:\Lambda^m_0\to\bm{\mathrm{sOb}}_\bullet(M)\) be a horn with missing
	\(0\)-th face.  Denote by \(\pi:\bm{\mathrm{sOb}}_\bullet(M)\to\bm{\mathrm{Sing}}^\infty_\bullet(M)\)
	the projection that forgets the Hamiltonian form.  Since \(\bm{\mathrm{Sing}}^\infty_\bullet(M)\)
	is a Kan complex, there exists a smooth singular \(m\)-simplex \(\sigma:\Delta^m\to M\) such
	that \(\pi\circ f = \sigma|_{\Lambda^m_0}\).  Hence for each \(i=1,\dots,m\) there is a
	Hamiltonian \((m-1)\)-form \(\alpha_i\) defined on a contractible open neighbourhood \(U\) of
	\(\sigma(\Delta^m)\) satisfying
	\[
	f_i = (\sigma\circ\partial_i)^*\alpha_i \in \bm{\mathrm{sOb}}_{m-1}(M),
	\]
	and by definition each \(\alpha_i\) comes with a set of \(n-m+1\) mutually commuting
	Hamiltonian vector fields that generate its exterior derivative.
	
	\medskip
	\noindent\textbf{Step 1: Common auxiliary vector fields from edge compatibility.}
	For any two faces \(i,j\in\{1,\dots,m\}\) with \(i<j\), their common \((m-2)\)-dimensional edge in the horn can be expressed in two ways: as the \(i\)-th face of \(f_j\), or as the \((j-1)\)-st face of \(f_i\).  The simplicial identities of the horn force these two \((m-2)\)-simplices to coincide modulo closed forms.  Taking exterior derivatives and using the definition of the face maps in \(\bm{\mathrm{sOb}}_\bullet(M)\), we obtain
	\[
	-(-1)^i\,\iota_{v_i^{(j)}}\,d\alpha_j \;=\; -(-1)^{j-1}\,\iota_{v_{j-1}^{(i)}}\,d\alpha_i
	\qquad\text{on the common edge},
	\]
	where \(v_i^{(j)}\) is the inward unit normal vector field to \(\partial_i\Delta^{m-1}\) within the face \(\partial_j\Delta^m\), and \(v_{j-1}^{(i)}\) is the inward unit normal vector field to \(\partial_{j-1}\Delta^{m-1}\) within the face \(\partial_i\Delta^m\).  
	
	Substituting \(d\alpha_k = -\iota_{V_k}\omega\) for each face, the equality becomes
	\[
	-(-1)^i\,\iota_{v_i^{(j)}}\,\iota_{V_j}\omega \;=\; -(-1)^{j-1}\,\iota_{v_{j-1}^{(i)}}\,\iota_{V_i}\omega .
	\]
	Using  \(\iota_{v}\iota_{V} = \iota_{V\wedge v}\), we obtain
	\[
	-(-1)^i\iota_{ V_j\wedge v_i^{(j)}}\,\omega \;=-(-1)^{j-1}\; \iota_{ V_i\wedge v_{j-1}^{(i)}}\,\omega
	\qquad\text{on the common edge}.
	\]
	By the non-degeneracy of \(\omega\) and because the generating vector fields are pushforwards of the coordinate vector fields on the auxiliary \([0,1]^{n-m+1}\) directions, the equality of the interior products implies that the multivector fields \(V_i\wedge v_{j-1}^{(i)}\) and \(V_j\wedge v_i^{(j)}\) coincide modulo \(\operatorname{Ker}(T^{n-m+1})\) (see Section~\ref{s2.4}). Since this holds for every pair of faces, we can choose representatives such that all faces share the same \(n-m\) auxiliary Hamiltonian vector fields, denoted by \(u_1,\dots,u_{n-m}\).  Writing \(w_k\) for the pushforward of the inward unit normal vector field of \(\partial_k\Delta^m\) (which is the remaining generator of \(\alpha_k\)), we obtain the uniform expression
	\[
	d\alpha_k = (-1)^k\,\iota_{u_1\wedge\cdots\wedge u_{n-m}\wedge w_k}\,\omega|_U \qquad (k=1,\dots,m).
	\]
	The sign \((-1)^k\) is forced by the edge compatibility condition; it compensates the Koszul signs that arise from the ordering of the face normals in the horn data.
	
	\medskip
	\noindent\textbf{Step 2: Construction of the filler.}
	The multivector field \(u_1\wedge\cdots\wedge u_{n-m}\) is Hamiltonian because all \(u_j\) are Hamiltonian and pairwise commuting.  Hence there exists a Hamiltonian \(m\)-form \(\beta\) on \(U\), unique up to closed forms, such that
	\[
	d\beta = -\iota_{u_1\wedge\cdots\wedge u_{n-m}}\,\omega|_U .
	\]
	Set \(\tilde f(\sigma) = \sigma^*\beta \in \bm{\mathrm{sOb}}_m(M)\).
	
	\medskip
	\noindent\textbf{Step 3: Verification of the face restrictions.}
	For each \(k=1,\dots,m\), let \(\tilde v_k\) be the pushforward along \(\sigma\) of the inward unit normal vector field of \(\partial_k\Delta^m\).  By construction, \(\tilde v_k = w_k\). The face map \(d^k\) applied to \(\tilde f(\sigma)\) requires solving
	\[
	d\eta = -(-1)^k\,\iota_{\tilde v_k}\,d\beta .
	\]
	We compute the right-hand side.  Using the definition of \(\beta\) and the commutativity of the generating vector fields,
	\[
	\iota_{\tilde v_k}\,d\beta
	= -\iota_{\tilde v_k}\iota_{u_1\wedge\cdots\wedge u_{n-m}}\omega
	= -\iota_{u_1\wedge\cdots\wedge u_{n-m}\wedge\tilde v_k}\,\omega
	= -\iota_{u_1\wedge\cdots\wedge u_{n-m}\wedge w_k}\,\omega .
	\]
	Now, from Step 1 we have \(\iota_{u_1\wedge\cdots\wedge u_{n-m}\wedge w_k}\,\omega = (-1)^k\,d\alpha_k\). Substituting this,
	\[
	\iota_{\tilde v_k}\,d\beta =-(-1)^k\,d\alpha_k .
	\]
	Therefore the face map equation becomes
	\[
	d\eta=d\alpha_k .
	\]
	Thus we may choose \(\eta = \alpha_k\), and obtain
	\[
	d^k(\tilde f(\sigma)) = (\sigma\circ\partial_k)^*\eta = (\sigma\circ\partial_k)^*\alpha_k = f_k .
	\]
	
	Therefore \(\tilde f\) fills the horn \(f\).  Since \(f\) was arbitrary,
	\(\bm{\mathrm{sOb}}_\bullet(M)\) satisfies the Kan condition.
\end{proof}

\begin{remark}
	We have explicitly constructed only the semi-simplicial set $\bm{\mathrm{sOb}}_{\le n}(M)$ up to dimension $n$, and verified the face-face identities as well as the filling property for all horns of dimension $\le n$. This suffices to define an $n$-groupoid (i.e., an $(n,0)$-category), and all subsequent constructions (extended field theory, cohomology, pre-$n$-Hilbert space, polarizations) involve only simplices of dimension $\le n$; hence the $n$-truncated version is already complete.
	
	If one wishes to obtain an $\infty$-groupoid defined in all dimensions, one may freely generate a semi-simplicial set from the $n$-truncated version by left Kan extension, where higher simplices are all degenerate. By the Theorem~\ref{thm:unique-deg}, this semi-simplicial set satisfies the Kan condition and thus becomes a Kan complex. This extension procedure is standard and does not alter the low-dimensional data. Therefore, without risk of confusion, we continue to regard $\bm{\mathrm{sOb}}_\bullet(M)$ as an $\infty$-groupoid and use the language of simplicial sets. The main results of the present paper do not depend on the explicit construction of the higher-dimensional data.
\end{remark}

With this simplicial model of observables at hand, we can now formulate extended field theories in analogy with state-sum constructions. Just as in state-sum constructions for topological quantum field theory one assigns topological charges to a triangulation, we define an extended field theory by assigning observables to the singular simplicial complex of spacetime.

\begin{definition}
	An $n$-dimensional extended field theory on $\Sigma$ valued in $(M,\omega)$ is a simplicial map
	\[
	F: \bm{\mathrm{Sing}}^\infty_\bullet(\Sigma) \longrightarrow \bm{\mathrm{sOb}}_\bullet(M),
	\]
	where $\bm{\mathrm{Sing}}^\infty_\bullet(\Sigma)$ is the smooth singular simplicial set of $\Sigma$, and $\bm{\mathrm{sOb}}_\bullet(M)$ is the simplicial set of observables on $(M,\omega)$ constructed above. The functor sends each singular simplex $\sigma$ to its assigned Hamiltonian form $\eta_\sigma$.
\end{definition}

The simplicial set $\bm{\mathrm{sOb}}_\bullet(M)$ can be viewed as the moduli stack of extended field theories on $\Sigma$. Indeed, homotopy classes of simplicial maps $F:\bm{\mathrm{Sing}}^\infty_\bullet(\Sigma)\to\bm{\mathrm{sOb}}_\bullet(M)$ classify equivalence classes of theories that can be continuously deformed into one another. Thus $\bm{\mathrm{sOb}}_\bullet(M)$ itself plays the role of a ``moduli space of theories," and its geometric realization $|\bm{\mathrm{sOb}}_\bullet(M)|$ becomes the classifying space of the theory. Topological information during deformations is captured by the homotopy classes of the simplicial map $F$. In particular, theories with a symmetry group $G$ should be described by $\bm{\mathrm{B}}G$-twisted homotopy theory. Consequently, studying equivalence classes of theories valued in $M$ essentially amounts to studying the homotopy type of the simplicial set $\bm{\mathrm{sOb}}_\bullet(M)$ and its geometric realization. This viewpoint is of fundamental importance for the classification of physical phases: phase transitions between different physical phases are often accompanied by jumps in homotopy invariants, and the homotopy type of $\bm{\mathrm{sOb}}_\bullet(M)$ provides a precise mathematical characterization of such jumps.

In the discussion above, if one interprets the exterior derivative of a $k$-form Hamiltonian as a $k$-plectic structure, one obtains a field theory on a worldvolume modeled by a smooth singular $k$-simplex. In this sense, an $n$-dimensional field theory encoded by an $n$-plectic manifold $(M,\omega)$ naturally encodes transformations between its boundary $(n-1)$-dimensional field theories, which themselves are described by an $(n-1)$-plectic structure of the form $d\alpha$, where $\alpha$ is an $(n-1)$-form Hamiltonian on $(M,\omega)$. Consequently, symmetries of a given field theory are realized as field theories in one dimension higher, a perspective further developed in \cite{FMT2024, KLW}.

\subsubsection{Morphisms and the $n$-vector space structure}\label{s412}

Geometrically, a morphism in $\bm{\mathrm{sOb}}_\bullet(M)$ is represented by a $1$-simplex, encoding a relation between observables. More systematically, such morphisms are organized by the simplicial set $\bm{\mathrm{sOb}}_\bullet(\bm{\mathrm{DiffMan}}(\Delta^1,M))$. By the natural exponential adjunction
\[
\bm{\mathrm{DiffMan}}(\Delta^k,\bm{\mathrm{DiffMan}}(\Delta^1,M)) \;\cong\; \bm{\mathrm{DiffMan}}(\Delta^k\times\Delta^1,M),
\]
we define $\bm{\mathrm{sOb}}_\bullet(\bm{\mathrm{DiffMan}}(\Delta^1,M))$ by
\[
\bm{\mathrm{sOb}}_k(\bm{\mathrm{DiffMan}}(\Delta^1,M)) \;:=\;
\left\{ \tau_k^*\alpha \;\Bigg|\;
\begin{aligned}
	&\tau_k : \Delta^k \times \Delta^1 \to M \text{ is a smooth map},\\
	&\alpha \in \Omega_{\mathrm{Ham}}^{k+1}(M)
\end{aligned}
\right\},
\]
where $\tau_k^*\alpha$ denotes the pullback of $\alpha$ along $\tau_k$, regarded as a $(k+1)$-form on $\Delta^k \times \Delta^1$. The simplicial structure on $\bm{\mathrm{sOb}}_\bullet(\bm{\mathrm{DiffMan}}(\Delta^1,M))$ is defined analogously to that of $\bm{\mathrm{sOb}}_\bullet(M)$.

Observe that $\Delta^k\times\Delta^1$ is canonically isomorphic to the $(k+1)$-simplex $\Delta^{k+1}$. Under such an isomorphism, a smooth singular $k$-simplex $\tau_k:\Delta^k\times\Delta^1\to M$ corresponds uniquely to a smooth singular $(k+1)$-simplex $\sigma_{k+1}:\Delta^{k+1}\to M$. For a fixed $t\in\Delta^1$, the restriction $\tau_k|_{\Delta^k\times\{t\}}$ corresponds, via the same isomorphism, to the restriction of $\sigma_{k+1}$ to the $k$-dimensional face of $\Delta^{k+1}$ corresponding to $\Delta^k\times\{t\}$. Hence the pullback $\tau_k^*\alpha$ encodes a one-parameter family of $(k+1)$-forms on $\Delta^k$, each given by $(\tau_k|_{\Delta^k\times\{t\}})^*\alpha$.

More systematically, there is a bijection
\[
\bm{\mathrm{sOb}}_k(\bm{\mathrm{DiffMan}}(\Delta^1,M)) \;\cong\; \bm{\mathrm{sOb}}_{k+1}(M),
\]
which follows from the standard exponential adjunction \(\Delta^k\times\Delta^1 \cong \Delta^{k+1}\) together with the fact that the auxiliary directions are simply rearranged under this identification.  This bijection reflects the geometric fact that a \(k\)-parameter family of paths in \(M\) is equivalent to a \((k+1)\)-dimensional singular simplex.  In particular, \(\bm{\mathrm{sOb}}_\bullet(\bm{\mathrm{DiffMan}}(\Delta^1,M))\) is a shifted version of \(\bm{\mathrm{sOb}}_{\bullet+1}(M)\), a fact used implicitly throughout.

A $k$-simplex $\tau_k^*\beta$ in $\bm{\mathrm{sOb}}_\bullet(\bm{\mathrm{DiffMan}}(\Delta^1,M))$ determines a $k$-simplex in $\bm{\mathrm{sOb}}_\bullet(M)$ by restriction to a fixed $t\in\Delta^1$; in particular, restricting to $t=0$ and $t=1$ yields two $k$-simplices $\alpha_0,\alpha_1\in\bm{\mathrm{sOb}}_\bullet(M)$. The relationship between $\beta$ and these endpoint simplices is captured by the differential condition
\[
d\alpha_t = -\iota_v d\beta,
\]
where $v_t$ is a vector field tangent to the path $(\tau_k)_a:\Delta^1\to M$ (for a fixed $a\in\Delta^k$), and $\alpha_t$ denotes the restriction of $\beta$ at $t$. More concretely, a $k$-simplex $\sigma_k^*\alpha$ in $\bm{\mathrm{sOb}}_\bullet(M)$ is related to a $k$-simplex $\tau_k^*\beta$ in $\bm{\mathrm{sOb}}_\bullet(\bm{\mathrm{DiffMan}}(\Delta^1,M))$ via this equation, reflecting the idea that a morphism corresponds to a continuous family of observables parametrized by a path in $M$, or equivalently, to the time evolution of an observable.

\begin{remark}
	By definition, $\bm{\mathrm{sOb}}_\bullet(M)$ is an $(n-1)$-truncated semi-simplicial set: the only non-degenerate simplex in degree $n$ is $\sigma^*\theta$, and there are no non-degenerate simplices above degree $n$. Morphisms in this setting are encoded by the path-space semi-simplicial set $\bm{\mathrm{sOb}}_\bullet(\bm{\mathrm{DiffMan}}(\Delta^1,M))$, which is naturally isomorphic to $\bm{\mathrm{sOb}}_{\bullet+1}(M)$ and is therefore $(n-2)$-truncated.
	
	Theorem~\ref{s4t1} establishes that $\bm{\mathrm{sOb}}_{\bullet+1}(M)$ satisfies the Kan filling condition and thus forms a semi-Kan complex. Its linearization $\mathbb{R}[\bm{\mathrm{sOb}}_{\bullet+1}(M)]$ carries a natural monoidal structure inherited from the composition of paths. Consequently, the linear span $\mathbb{R}[\bm{\mathrm{sOb}}_\bullet(M)]$ can be regarded as a category enriched over $\mathbb{R}[\bm{\mathrm{sOb}}_{\bullet+1}(M)]$, which is precisely the defining structure of an $n$-vector space. In fact, the $n$-vector space obtained in this way is the dual of the one constructed in Section~\ref{s42}.
\end{remark}

\subsubsection{Cohomological invariants and higher structures}\label{s413}
As an integral part of this simplicial homotopy framework, we can extract invariants of observables via cohomology. The recursive gluing conditions are encoded in the face maps of $\bm{\mathrm{sOb}}_\bullet(M)$. Consequently, the obstructions to gluing---i.e., the failure of local Hamiltonian forms to patch into a global object---are captured by the cohomology of a naturally defined cochain complex.

\begin{definition}[Chain and cochain complexes of observables]
	First define the chain complex $\bm{\mathrm{C}}_\bullet^{\mathrm{obs}}(M)$ by
	\[
	\bm{\mathrm{C}}_k^{\mathrm{obs}}(M) := \mathbb{Z}[\bm{\mathrm{sOb}}_k(M)],
	\]
	with boundary operator $\partial_k: \bm{\mathrm{C}}_k^{\mathrm{obs}}(M) \to \bm{\mathrm{C}}_{k-1}^{\mathrm{obs}}(M)$ given by
	\[
	\partial_k(\sigma) = \sum_{i=0}^{k} (-1)^i d_i\sigma,
	\]
	where $d_i$ are the face maps of $\bm{\mathrm{sOb}}_\bullet(M)$. The homology of this complex is denoted $H_\bullet^{\mathrm{obs}}(M)$.
	
	Dually, define the cochain complex $\bm{\mathrm{Ch}}^\bullet_{\mathrm{obs}}(M)$ by
	\[
	\bm{\mathrm{Ch}}^k_{\mathrm{obs}}(M) := \mathrm{Hom}_{\mathbb{Z}}(\mathbb{Z}[\bm{\mathrm{sOb}}_k(M)], \mathbb{K}),
	\]
	with coboundary operator $\delta^k: \bm{\mathrm{Ch}}^k_{\mathrm{obs}}(M) \to \bm{\mathrm{Ch}}^{k+1}_{\mathrm{obs}}(M)$ given by
	\[
	(\delta^k f)(\sigma) = \sum_{i=0}^{k+1} (-1)^i f(d_i\sigma).
	\]
	The cohomology of this complex is denoted $H^\bullet_{\mathrm{obs}}(M)$.
	
	By construction, $\bm{\mathrm{Ch}}^\bullet_{\mathrm{obs}}(M)$ and $\bm{\mathrm{C}}_\bullet^{\mathrm{obs}}(M)$ are linear duals of each other, and the coboundary operator $\delta$ satisfies $\delta f(\sigma) = f(\partial \sigma)$.
\end{definition}

The cohomology $H^k_{\mathrm{obs}}(M)$ of this complex thus serves as a candidate for encoding obstructions to gluing Hamiltonian forms. These cochains can also be interpreted as observables valued in $\mathbb{K}$; then the cohomology captures information invariant under continuous deformation and can be used to classify observables. Indeed, the degree-$n$ cohomology corresponds precisely to the adiabatic invariant $\int_{\Delta^n} \phi^*\theta$, where $\phi: \Delta^n \to M$ and $d\theta = \omega$.

Beyond the cohomology groups themselves, the cochain complex $(\bm{\mathrm{Ch}}^\bullet_{\mathrm{obs}}(M), \delta)$ carries a richer algebraic structure: the cup product endows it with a natural $E_\infty$-algebra structure, reflecting the higher homotopy coherence of the product on cochains. This structure is expected to capture finer quantum information, such as the braiding statistics of quasiparticles or extended objects.

Finally, we turn to the interpretation of quantum states, which will be discussed within the framework of Section~\ref{s42}. Quantum states are defined via $\mathrm{Map}(\bm{\mathrm{sOb}}_\bullet(M),U(1))$, i.e., as complex-valued cochains on the simplicial set of observables. In this picture, the cohomology of $\bm{\mathrm{Ch}}^\bullet_{\mathrm{obs}}(M)$ classifies the continuous deformation classes of quantum states: a non-trivial cohomology class corresponds to a topological excitation, while the information of topological phases is encoded in the algebraic structure of the cocycles. This provides a unified cohomological perspective on both classical obstructions and quantum topological phenomena.

\subsubsection{Relation to $L_\infty$-algebras}
\label{s44}

The semi-simplicial set $\bm{\mathrm{sOb}}_\bullet(M)$ faithfully represents the totality of Hamiltonian forms on $M$. However, its face maps naturally close only on the subalgebra generated by mutually commuting Hamiltonian vector fields, i.e., the higher Heisenberg algebra of Remark~\ref{rem:heisenberg}. The full non-abelian $L_\infty$-algebra $L$ constructed in Theorem~\ref{thmmain} should therefore be interpreted as the symmetry algebra that is expected to act on $\bm{\mathrm{sOb}}_\bullet(M)$, rather than as a structure directly carried by the simplicial set itself.

A complete realisation of this action encounters a fundamental geometric obstruction. Each $k$-simplex of $\bm{\mathrm{sOb}}_\bullet(M)$ consists of two intertwined data: an underlying singular simplex $\sigma_k$ and a Hamiltonian form $\alpha$. While the action on the Hamiltonian form can be expressed algebraically via the higher Poisson brackets, the action on the singular simplex would require a geometric deformation of $\sigma_k$ along the Hamiltonian multivector field $v_x$ associated with an element $x\in L$, producing a new singular simplex whose dimension reflects the degree shift of the bracket. This geometric deformation cannot be captured by the purely combinatorial face maps of the semi-simplicial set and remains an open problem. A systematic construction of such an $L_\infty$-module structure on $\bm{\mathrm{sOb}}_\bullet(M)$, and its dual action on the quantum state space, is left for future investigation.

\section{Relation to $n$-gerbes  and quantization}
\label{s5}

\subsection{Simplicial observables as a combinatorial $(n-1)$-gerbe}

An $n$-gerbe with connection on $M$ with band $U(1)$ is described by \v{C}ech--Deligne data relative to a good open cover $\{U_i\}$. For an $n$-plectic manifold $(M,\omega)$, the integrality condition $[\omega/2\pi]\in H^{n+1}(M,\mathbb{Z})$ is precisely the condition for the existence of such a gerbe whose curvature is $\omega$. We now show that, under this prequantum condition, $\bm{\mathrm{sOb}}_\bullet(M)$ provides a combinatorial model of the corresponding $n$-gerbe. We illustrate this for $n=2$.

Take a good open cover $\{U_i\}$ of $M$ with $Y = \coprod_i U_i$. The $0$-, $1$-, and $2$-simplices of $\bm{\mathrm{sOb}}_\bullet(M)$ encode the local gerbe data as follows:
\begin{itemize}
	\item $0$-simplex $\sigma_0^*\alpha_0$: $\sigma_0$ is a point in $U_i$, giving a point in the base space $Y$.
	\item $1$-simplex $\sigma_1^*\alpha_1$ with $\sigma_1: \Delta^1 \to U_{ij}$: the Hamiltonian $1$-form $\alpha_1$ ($d\alpha_1 = -\iota_{v_{\alpha_1}}\omega$) encodes the transition data on the double intersection. Its path integral
	\[
	g_{ij} = e^{i\int_{\Delta^1}\sigma_1^*\alpha_1}
	\]
	defines the $U(1)$-valued transition function on $U_{ij}$, which relates the local frames on $U_i$ and $U_j$ for the line bundle $L$ on $Y^{[2]}$.
	\item $2$-simplex $\sigma_2^*\theta$ with $\sigma_2: \Delta^2 \to U_{ijk}$ and $d\theta = \omega$: the local primitive $\theta$ defines the $2$-cocycle
	\[
	c_{ijk} = e^{i\int_{\Delta^2}\sigma_2^*\theta}.
	\]
	(When verifying that $c_{ijk}$ is well-defined, one introduces local primitives $\theta_i$ on each $U_i$ and transition $1$-forms $\lambda_{ij}$ with $\theta_j - \theta_i = d\lambda_{ij}$; the integrality condition then guarantees independence of all choices.)
\end{itemize}
The Kan filler for the horn formed by $g_{ij}$ and $g_{jk}$ produces a $2$-simplex $\sigma_2^*\theta$; from this filler we obtain a $U(1)$-valued function $c_{ijk}=e^{i\int_{\Delta^2}\sigma_2^*\theta}$ and a third edge $g_{ik}$ (the image of the face opposite the vertex $j$). We can view the filler itself as a mapping that sends the ordered pair $(g_{ij},g_{jk})$ to $g_{ik}$. This mapping, denoted by the same symbol $c_{ijk}$, is defined by the Kan construction:
\[
c_{ijk}(g_{ij},g_{jk}) = g_{ik},
\]
where the right-hand side is the transition function extracted from the third face of the filler. In this picture, this means the multiplication map $c_{ijk}: L_{ij}\otimes L_{jk}\to L_{ik}$ on triple intersections satisfies $c_{ijk}(g_{ij}\otimes g_{jk}) =  g_{ik}$ for any local sections $g_{ij}, g_{jk}, g_{ik}$. Thus the Kan filling property directly yields the gerbe multiplication.

For these data to constitute a gerbe, the transition functions must satisfy the associativity (cocycle) condition on each quadruple intersection $U_{ijkl}$:
\[
c_{jkl}\, c_{ikl}^{-1}\, c_{ijl}\, c_{ijk}^{-1} = 1.
\]
The face maps of $\bm{\mathrm{sOb}}_\bullet(M)$ realise the restriction of a $2$-simplex to its three boundary $1$-simplices, so that the above condition is encoded in the simplicial structure. By Stokes' theorem,
\[
c_{jkl}\, c_{ikl}^{-1}\, c_{ijl}\, c_{ijk}^{-1}
= e^{i\int_{\partial\Delta^3}\sigma_3^*\theta}
= e^{i\int_{\Delta^3}\sigma_3^*\omega}.
\]
Thus the associativity condition holds for all $3$-simplices if and only if $\int_{\Delta^3}\sigma_3^*\omega \in 2\pi\mathbb{Z}$ for every $3$-simplex $\sigma_3$. Since any closed $3$-chain can be represented by integral combinations of $3$-simplices, this is equivalent to the prequantum condition $[\omega/2\pi]\in H^3(M,\mathbb{Z})$. Hence the prequantum condition is necessary for $\bm{\mathrm{sOb}}_\bullet(M)$ to define a gerbe.

Conversely, assume $[\omega/2\pi]\in H^3(M,\mathbb{Z})$. Choose contractible $U_i$ and local primitives $\theta_i$ with $d\theta_i = \omega|_{U_i}$. On $U_{ij}$, pick $\lambda_{ij}$ with $d\lambda_{ij} = \theta_j - \theta_i$. For any closed $2$-chain $S \subset U_{ij}$, take $3$-chains $B_i \subset U_i$, $B_j \subset U_j$ with $\partial B_i = \partial B_j = S$. Applying the prequantum condition to the closed $3$-chain $B_i - B_j$ gives
\[
\int_S d\lambda_{ij} = \int_S (\theta_j - \theta_i) = \int_{B_j}\omega - \int_{B_i}\omega \in 2\pi\mathbb{Z}.
\]
Hence $c_{ijk} = e^{i\int_{\Delta^2}\sigma_2^*\theta}$ is a well-defined $U(1)$-valued function depending only on the homotopy class of $\sigma_2$. A standard \v{C}ech--de Rham argument shows $\delta c = 1$ on $U_{ijkl}$, so the $\{c_{ijk}\}$ form a $2$-cocycle. Together with $\{\theta_i\}$ and $\{\lambda_{ij}\}$, these data constitute a bundle gerbe with connection whose curvature is $\omega$. Thus the prequantum condition is also sufficient.

For general $n$, the $k$-simplices of $\bm{\mathrm{sOb}}_\bullet(M)$ ($0 \le k \le n$) supply the local $k$-form components of the \v{C}ech--Deligne data of an $n$-gerbe with connection. The Kan property (Theorem~\ref{s4t1}) guarantees the coherence of all higher cocycle conditions, and the prequantum condition $[\omega/2\pi]\in H^{n+1}(M,\mathbb{Z})$ is equivalent to the closedness of the top-degree \v{C}ech--Deligne cocycle.

\subsection{Pre-$n$-Hilbert spaces}\label{s42}

In symplectic geometry, sections of the prequantum line bundle form a pre-Hilbert space. In simplicial terms, $0$-cochains give the functions, while $1$-cochains provide the phase $e^{i\int\theta}$ along paths. For an $n$-plectic manifold $(M,\omega)$, we encode quantum states as follows.

\begin{definition}[Cosimplicial set of quantum states]
	Let $\bm{\mathrm{sOb}}_\bullet(M)$ be the simplicial set of observables. Define the cosimplicial set $\bm{\mathrm{qOb}}^\bullet(M)$ by
	\[
	\bm{\mathrm{qOb}}^k(M) := \mathrm{Map}(\bm{\mathrm{sOb}}_k(M), U(1)),
	\]
	with cosimplicial structure induced from $\bm{\mathrm{sOb}}_\bullet(M)$.
\end{definition}

The $U(1)$-module $U(1)[\bm{\mathrm{qOb}}^\bullet(M)]$, equivalently the cochain complex $\bm{\mathrm{Ch}}^\bullet_{\mathrm{obs}}(M)$ of Section~\ref{s413}, is the pre-$n$-Hilbert space. Its $n$-vector space structure is the quantum counterpart of Section~\ref{s412}.

To introduce a categorified inner product, we follow the analogy with quantum mechanics. There, the inner product on a symplectic manifold $(M,\omega)$ is the phase space path integral
\[
\langle \psi_f | \psi_i \rangle = \int_{M\times M}dxdy\int \mathcal{D}[\gamma] \, \overline{\psi_f(\gamma(1))} \, e^{i\int_0^1 \gamma^*\theta} \, \psi_i(\gamma(0)),
\]
where $\gamma$ is a path from $y$ to $x$ and $d\theta=\omega$. The factor $e^{i\int\theta}$ is a $U(1)$-valued $1$-cocycle on $\bm{\mathrm{sOb}}_\bullet(M)$, the cocycle condition following from Stokes' theorem and integrality of $\omega$.

For an $n$-plectic manifold, a path in the mapping space $\bm{\mathrm{DiffMan}}(\Delta^k,M)$ describes the evolution of a $k$-dimensional extended object. We take as integral kernel a $U(1)$-valued $1$-cocycle on the path space:
\[
e^{i\mathcal{K}} \in \bm{\mathrm{Ch}}_{\mathrm{obs}}^1\bigl(\bm{\mathrm{DiffMan}}(\Delta^1, \bm{\mathrm{DiffMan}}(\Delta^k,M)), U(1)\bigr), \qquad \delta e^{i\mathcal{K}} = 1.
\]
For $\psi_f^\bullet,\psi_i^\bullet \in \bm{\mathrm{qOb}}^\bullet(M)$, the inner product at level $k$ is the formal path integral
\[
\langle \psi_f^k | \psi_i^k \rangle_{e^{i\mathcal{K}}} := \int \mathcal{D}[\sigma_f^*\alpha_f] \mathcal{D}[\sigma_i^*\alpha_i] \int \mathcal{D}[\gamma] \; \overline{\psi_f^k(\sigma_f^*\alpha_f)} \, e^{i\mathcal{K}(\gamma)} \, \psi_i^k(\sigma_i^*\alpha_i),
\]
where integrals are formal sums over the indicated sets and $\gamma(0)=\sigma_f^*\alpha_f$, $\gamma(1)=\sigma_i^*\alpha_i$.

Via the adjunction $\bm{\mathrm{DiffMan}}(\Delta^1, \bm{\mathrm{DiffMan}}(\Delta^k,M)) \cong \bm{\mathrm{DiffMan}}(\Delta^{k+1},M)$, the cocycle $e^{i\mathcal{K}}$ pulls back to a $(k+1)$-cocycle on $\bm{\mathrm{sOb}}_\bullet(M)$, still denoted $e^{i\mathcal{K}}$, with $\delta e^{i\mathcal{K}} = 1$. The boundary condition becomes $\partial(\sigma_{k+1}^*\beta) = \sigma_f^*\alpha_f - \sigma_i^*\alpha_i$. The inner product rewrites as
\[
\langle \psi_f^k | \psi_i^k \rangle_{e^{i\mathcal{K}}} = \int \mathcal{D}[\sigma_f^*\alpha_f] \mathcal{D}[\sigma_i^*\alpha_i] \int \mathcal{D}[\sigma_{k+1}] \overline{\psi_f^k(\sigma_f^*\alpha_f)} \; e^{i\mathcal{K}(\sigma_{k+1}^*\beta)} \; \psi_i^k(\sigma_i^*\alpha_i).
\]

The cocycle condition forces $e^{i\mathcal{K}(\sigma_{k+1}^*\beta)}$ to depend only on the homotopy class of the interpolating simplex: any two such simplices with the same boundary can be capped by a $(k+2)$-simplex, and $\delta e^{i\mathcal{K}} = 1$ implies their phases coincide. The pairing takes values in $U(1)[\bm{\mathrm{qOb}}^{\bullet+1}(M)]$.

\begin{remark}[From formal sum to finite sum and the hierarchy of functional integrals]
	The inner sum over interpolating $(k+1)$-simplices collapses to a finite sum over homotopy classes because the integral kernel $e^{i\mathcal{K}}$ is a cocycle.  The outer integrals over the boundary states $\sigma_f^*\alpha_f$ and $\sigma_i^*\alpha_i$ are formal functional integrals.  Their meaning is best understood through the hierarchy of the theory: for $n=2$ and $k=0$ the boundary states are functions on points of $M$ and the inner product reduces to an ordinary integral; for $n=2$ and $k=1$ the boundary states are functions on paths and the outer integrals become functional integrals over $\bm{\mathrm{DiffMan}}(\Delta^1,M)$, analogous to the inner product of states in the geometric quantisation of loop spaces.  In the latter setting, such functional integrals are typically made rigorous by constructing a measure compatible with the symplectic structure (e.g., via a Wiener measure in the K\"ahler case, or heat kernel measures on compact Lie groups).  The present simplicial framework is a smooth model; when it is restricted to a finite triangulation of $M$, the formal functional integrals reduce to genuine finite combinatorial sums, yielding a well-defined complex number without any auxiliary metric data.
\end{remark}

\begin{remark}[Anomaly cancellation, Wess-Zumino terms, and Wilson membranes]
	A \(k\)-dimensional boundary theory described by a quantum state \(\psi^k\) is generically anomalous.  As the boundary of a \((k+1)\)-dimensional bulk, the boundary state \(\psi^k\) is obtained by fixing a bulk extension of the background gerbe data.  Different choices of the bulk extension correspond precisely to different gauge choices for \(\psi^k\); under a change of the bulk extension the boundary state transforms by a multiplicative phase, which is the signature of an anomaly.  The cocycle \(e^{i\mathcal{K}}\) entering the inner product is the Wess-Zumino term that compensates this anomaly via the anomaly inflow mechanism: its variation with respect to the bulk extension cancels the variation of the boundary state, so that the combined bulk-boundary amplitude is gauge-invariant.  The cocycle condition \(\delta e^{i\mathcal{K}} = 1\) is the algebraic statement of this invariance.  Equivalently, \(e^{i\mathcal{K}}\) is a Wilson membrane whose endpoints are the charged quantum states \(\psi^k\), and the inner product \(\langle \psi_f^k \mid \psi_i^k \rangle_{e^{i\mathcal{K}}}\) is the anomaly-free amplitude for this Wilson membrane to interpolate between the two boundary configurations.
\end{remark}

\begin{remark}[Relation to gerbe-section \(2\)-Hilbert spaces and higher-rank generalisations]
	In the \(2\)-plectic case, Bunk--S\"amann--Szab\'o~\cite{BS} defined a \(2\)-Hilbert space whose objects are twisted vector bundles of arbitrary finite rank over the prequantum gerbe.  The cosimplicial space \(U(1)[\bm{\mathrm{qOb}}^\bullet(M)]\) introduced in the present paper encodes, by contrast, the local sections of twisted line bundles over the same gerbe---the rank-\(1\) subsector of their construction.  This is the natural higher categorical analogue of the fact that, in ordinary geometric quantisation, quantum states are sections of the prequantum line bundle rather than the bundle itself.
	
	To promote this to higher-rank twisted vector bundles, and in particular to the \(2\)-vector bundles of Kristel, Ludewig and Waldorf~\cite{Wal2021}, one may replace the coefficient object \(U(1)\) by the symmetric monoidal \(2\)-category \(\bm{\mathrm{Alg}}\) whose objects are associative algebras, \(1\)-morphisms are bimodules, and \(2\)-morphisms are intertwiners.  The cosimplicial set \(\bm{\mathrm{qOb}}^\bullet(M)\) then becomes a cosimplicial object in this \(2\)-category, and its linearisation can be understood as a combinatorial model for the \(2\)-vector space of sections of a \(2\)-line bundle.  A detailed investigation of this enriched setting is left for future work.
\end{remark}

\subsection{Hierarchy of polarizations}\label{s43}

The categorified inner product introduced in Section~\ref{s42} is governed by a cocycle $e^{i\mathcal{K}}$ on the space of paths in the mapping space $\mathcal{P}_k := \bm{\mathrm{DiffMan}}(\Delta^k,M)$. We now explain how this algebraic datum induces a natural symplectic structure on $\mathcal{P}_k$, and how the cocycle condition translates into a prequantum condition for that symplectic form. This provides the geometric foundation for a recursive polarization scheme that systematically reduces the pre-$n$-Hilbert space to the physical $n$-Hilbert space.

The variation of the functional $\mathcal{K}$ with respect to a path defines a $1$-form $\Theta_{\mathcal{K}}$ on $\mathcal{P}_k$. For a point $\sigma \in \mathcal{P}_k$ and a tangent vector $X \in T_\sigma \mathcal{P}_k$, let $\gamma_\varepsilon$ be a path in $\mathcal{P}_k$ with $\gamma_0 = \sigma$ and $\left.\frac{d}{d\varepsilon}\right|_{\varepsilon=0} \gamma_\varepsilon = X$. Set
\[
\Theta_{\mathcal{K}}|_{\sigma}(X) := \left. \frac{d}{d\varepsilon} \right|_{\varepsilon=0} \mathcal{K}(\gamma_\varepsilon).
\]
Its exterior derivative $\Omega_{\mathcal{K}} := d \Theta_{\mathcal{K}} \in \Omega^2(\mathcal{P}_k)$ is a symplectic form. When $\mathcal{K}$ arises from a Hamiltonian form $\beta$ on $M$, $\Theta_{\mathcal{K}}$ is the integral of $\beta$ over $\Delta^k$, and $\Omega_{\mathcal{K}}$ reduces to the transgressed form $\int_{\Delta^k} \mathrm{ev}_k^* d\beta$.

The algebraic cocycle condition $\delta e^{i\mathcal{K}} = 1$ on the path space has a direct geometric interpretation on $\mathcal{P}_k$. For any closed loop $\ell$ in $\mathcal{P}_k$, the cocycle condition applied to a bounding $2$-simplex implies $\exp( i \oint_\ell \Theta_{\mathcal{K}} ) = 1$, which is equivalent to the integrality of $\Omega_{\mathcal{K}}$:
\[
[\Omega_{\mathcal{K}} / 2\pi] \in H^2(\mathcal{P}_k, \mathbb{Z}).
\]
Thus the choice of a cocycle as the integral kernel automatically makes each $(\mathcal{P}_k, \Omega_{\mathcal{K}})$ a prequantizable symplectic manifold.

Since $(\mathcal{P}_k, \Omega_{\mathcal{K}})$ is prequantizable, we may apply geometric quantization. A polarization is a Lagrangian submanifold $\mathcal{L}_k \subset \mathcal{P}_k$; the physical quantum state space consists of elements of $U(1)[\bm{\mathrm{qOb}}^k(M)]$ covariantly constant transverse to $\mathcal{L}_k$. The inner product pairs states in the same Lagrangian submanifold, reducing the prequantum space to the physical state space.

The hierarchy proceeds recursively from the top dimension downward. At level $k=n$, $U(1)[\bm{\mathrm{qOb}}^n(M)]$ is interpreted as the partition function on $\Delta^n$, while $U(1)[\bm{\mathrm{qOb}}^{n-1}(M)]$ gives the quantum states on its boundary. A polarization on $\mathcal{P}_{n-1}$ yields the wave function space familiar from ordinary field theory quantization. Repeating this from $k=n$ down to $0$, each step performs a polarization: one conjugate momentum direction is eliminated, and the transverse motion becomes a new positional degree of freedom. The form degree decreases from $k$ to $k-1$, capturing both dimensional reduction and the redistribution of degrees of freedom.

The resulting hierarchy is compatible with the \(1\)-polarization classification of~\cite{r2011}, according to which an \(n\)-plectic manifold admits \(n\) distinct orders of polarization.  To see this, recall that in the De Donder--Weyl formulation of classical field theory, the multisymplectic potential locally takes the form \(\Theta = p^\mu \wedge d\phi \wedge d^{n-1}x_\mu\), where \(\phi\) denotes the field, \(p^\mu\) its conjugate polymomenta, and \(d^{n-1}x_\mu\) the volume element on a codimension-\(1\) hypersurface.  A \(1\)-polarization in this setting eliminates all polymomentum components at once, projecting the prequantum state space onto functionals of the field configuration \(\phi\) and its domain coordinates \(x\).  In our framework, the phase space is stratified over the hierarchy of mapping spaces, with each level carrying a single polymomentum component.  Applying a symplectic polarization at each level eliminates the \(p^\mu\) one by one; the cumulative effect is therefore equivalent to a \(1\)-polarization of~\cite{r2011}.  The same recursive pattern can also be interpreted as a multi-step condensation of topological excitations~\cite{K2014436,K2024}: at each step a class of topological charges is condensed, driving the system to the next lower level and suggesting that polarization may be regarded as a mathematical realization of topological condensation.

\section{Conclusion and Outlook}\label{s6}

We have developed a unified framework for observables in multisymplectic geometry. Algebraically, the introduction of a degree shifting Grassmann variable $u$ extends Rogers' $L_\infty$-algebra to Hamiltonian forms of all degrees, yielding a graded $L_\infty$-algebra (Theorem~\ref{thmmain}) whose brackets are defined via interior products with Hamiltonian multivector fields. Geometrically, interpreting $k$-form observables as $k$-dimensional defects leads to a semi-simplicial set $\bm{\mathrm{sOb}}_\bullet(M)$, which we prove is a Kan complex (Theorem~\ref{s4t1}) and thus provides an $n$-groupoid model for observables. Its linearization yields a combinatorial $n$-vector space, and the associated cochain complex $\bm{\mathrm{Ch}}^\bullet_{\mathrm{obs}}(M)$ encodes gluing obstructions as well as adiabatic invariants. Under the prequantum condition $[\omega/2\pi]\in H^{n+1}(M,\mathbb{Z})$, $\bm{\mathrm{sOb}}_\bullet(M)$ furnishes a combinatorial model of an $n$-gerbe with connection. Quantum states, defined as complex-valued cochains on $\bm{\mathrm{sOb}}_\bullet(M)$, form a cosimplicial set $\bm{\mathrm{qOb}}^\bullet(M)$, and a recursive inner product taking values in $U(1)[\bm{\mathrm{qOb}}^{\bullet+1}(M)]$ endows the linearized space with the structure of a categorified pre-$n$-Hilbert space. The hierarchical reduction of form degrees matches the $1$-polarization scheme of~\cite{r2011}. 

Several directions remain for future investigation. A detailed comparison with the $2$-Hilbert space of~\cite{BS} for $n=2$ are natural next steps. Moreover, the left $L_\infty$-module structure on $\bm{\mathrm{C}}_\bullet^{\mathrm{obs}}(M)$ indicates that the full observable algebra acts geometrically on $\bm{\mathrm{sOb}}_\bullet(M)$.  Finally, the $E_\infty$-algebra structure on $\bm{\mathrm{Ch}}^\bullet_{\mathrm{obs}}(M)$ is expected to capture braiding statistics of topological excitations, such as anyons in topological phases of matter. 

\section*{Acknowledgements}

The author is grateful to Xiang-Mao Ding for his guidance and support throughout this project. Special thanks are due to Hank Chen for his careful reading of the manuscript and for numerous insightful comments that substantially improved this work. His suggestions regarding the higher categorical structures relevant to the constructions in this paper, as well as his remarks on potential physical applications, provided valuable directions for refining the presentation and for future investigation. 

\bibliographystyle{plainurl} 
\bibliography{ref.bib} 

@article{HF,
	author = {Hélein, Frédéric},
	eprint={math-ph/0212036},
	archivePrefix={arXiv},
	primaryClass={math-ph},
	year = {2003},
	month = {01},
	pages = {},
	title = {Hamiltonian formalisms for multidimensional calculus of variations and perturbation theory},
	isbn = {9780821836354},
	doi = {10.1090/conm/350/06342}
}

@article{RN,
	author = {Román-Roy, Narciso},
	eprint={math-ph/0506022},
	archivePrefix={arXiv},
	primaryClass={math-ph},
	year = {2009},
	month = {11},
	pages = {},
	title = {{Multisymplectic Lagrangian and Hamiltonian formalisms of classical field theories}},
	journal = {Symmetry, Integrability and Geometry: Methods and Applications},
	doi = {10.3842/SIGMA.2009.100}
}

@article{Rogers_2011,
   title={{$L_\infty$-Algebras from Multisymplectic Geometry}},
   volume={100},
   ISSN={1573-0530},   
   DOI={10.1007/s11005-011-0493-x},
   number={1},
   journal={Letters in Mathematical Physics},
   publisher={Springer Science and Business Media LLC},
   author={Rogers, Christopher L.},
   year={2011},
   month=apr, pages={29–50},
   eprint={1005.2230},
   archivePrefix={arXiv},
   primaryClass={math-DG},}

@article{FORGER_2003,
   title={The Poisson Bracket for Poisson Forms in Multisymplectic Field Theory},
   volume={15},
   ISSN={1793-6659},   
   DOI={10.1142/s0129055x03001734},
   number={07},
   journal={Reviews in Mathematical Physics},
   publisher={World Scientific Pub Co Pte Lt},
   author={Forger, Michael and Paufler, Cornelius and Römer, Hartmann},
   year={2003},
   month=sep, pages={705–743},
   eprint={1608.08455},
   archivePrefix={arXiv},
   primaryClass={math-ph}, }

@article{Lada1994,
	author = {Tom Lada and Martin Markl},
	title = {{Strongly homotopy Lie algebras}},
	journal = {Communications in Algebra},
	volume = {23},
	number = {6},
	pages = {2147--2161},
	year = {1995},
	publisher = {Taylor \& Francis},
	doi = {10.1080/00927879508825335},
	eprint={hep-th/9406095},
	archivePrefix={arXiv},
	primaryClass={hep-th},
	
}

@article{lada1993,
   title={{Introduction to SH Lie algebras for physicists}},
   volume={32},
   ISSN={1572-9575},
   
   DOI={10.1007/bf00671791},
   number={7},
   journal={International Journal of Theoretical Physics},
   publisher={Springer Science and Business Media LLC},
   author={Lada, Tom and Stasheff, Jim},
   year={1993},
   month=jul, pages={1087–1103},
   eprint={hep-th/9209099},
   archivePrefix={arXiv},
   primaryClass={hep-th}, }

@article{Baez1995,
    author = "Baez, J. C. and Dolan, J.",
    title = "{Higher dimensional algebra and topological quantum field theory}",
    eprint = "q-alg/9503002",
    archivePrefix = "arXiv",
    doi = "10.1063/1.531236",
    journal = "J. Math. Phys.",
    volume = "36",
    pages = "6073--6105",
    year = "1995"
}

@article{Baez2009,
   title={Categorified Symplectic Geometry and the Classical String},
   eprint = "0808.0246",
   archivePrefix = "arXiv",
   primaryClass = "math-ph",
   volume={293},
   ISSN={1432-0916},  
   DOI={10.1007/s00220-009-0951-9},
   number={3},
   journal={Communications in Mathematical Physics},
   publisher={Springer Science and Business Media LLC},
   author={Baez, John C. and Hoffnung, Alexander E. and Rogers, Christopher L.},
   year={2009},
   month=nov, pages={701–725} }

@article{Baez2010ya,
    author = "Baez, John C. and Huerta, John",
    title = "{An Invitation to Higher Gauge Theory}",
    eprint = "1003.4485",
    archivePrefix = "arXiv",
    primaryClass = "hep-th",
    doi = "10.1007/s10714-010-1070-9",
    journal = "Gen. Rel. Grav.",
    volume = "43",
    pages = "2335--2392",
    year = "2011"
}

@incollection{Baez2005qu,
author = {Baez, John C. and Schreiber, Urs},
title = {Higher gauge theory},
booktitle = {Categories in Algebra, Geometry and Mathematical Physics},
series = {Contemp. Math.},
volume = {431},
pages = {7--30},
publisher = {American Mathematical Society},
address = {Providence, RI},
year = {2007},
isbn = {9780821839706},
doi = {10.1090/conm/431/08263},
editor = {Davydov, A. and others},
eprint = "math/0511710",
}

@article{Witten1983ar,
    author = "Witten, Edward",
    editor = "Stone, M.",
    title = "{Nonabelian Bosonization in Two-Dimensions}",
    reportNumber = "PRINT-83-0934 (PRINCETON)",
    doi = "10.1007/BF01215276",
    journal = "Commun. Math. Phys.",
    volume = "92",
    pages = "455--472",
    year = "1984"
}

@incollection{Lurie2009,
 author = {Lurie, Jacob},
 title = {On the classification of topological field theories},
 booktitle = {Current developments in mathematics, 2008},
 isbn = {978-1-57146-139-1},
 pages = {129--280},
 year = {2009},
 publisher = {Somerville, MA: International Press},
 language = {English},
 zbMATH = {5635200},
 Zbl = {1180.81122},
 doi="10.4310/CDM.2008.V2008.N1.A3",
 eprint={0905.0465},
 archivePrefix={arXiv},
 primaryClass={math.CT},
}

@article{B2024,
   title={{Generalized charges, part I: Invertible symmetries and higher representations}},
   eprint = "2304.02660",
   archivePrefix = "arXiv",
   primaryClass = "hep-th",
   volume={16},
   ISSN={2542-4653},   
   DOI={10.21468/scipostphys.16.4.093},
   number={4},
   journal={SciPost Physics},
   publisher={Stichting SciPost},
   author={Bhardwaj, Lakshya and Schäfer-Nameki, Sakura},
   year={2024},
   month=apr }

@article{B2025,
   title={{Generalized charges, part II: Non-invertible symmetries and the symmetry TFT}},
   eprint = "2305.17159",
   archivePrefix = "arXiv",
   primaryClass = "hep-th",
   volume={19},
   ISSN={2542-4653},  
   DOI={10.21468/scipostphys.19.4.098},
   number={4},
   journal={SciPost Physics},
   publisher={Stichting SciPost},
   author={Bhardwaj, Lakshya and Schäfer-Nameki, Sakura},
   year={2025},
   month=oct }

@article{G2014,
    author = "Gaiotto, Davide and Kapustin, Anton and Seiberg, Nathan and Willett, Brian",
    title = "{Generalized Global Symmetries}",
    eprint = "1412.5148",
    archivePrefix = "arXiv",
    primaryClass = "hep-th",
    doi = "10.1007/JHEP02(2015)172",
    journal = "JHEP",
    volume = "02",
    pages = "172",
    year = "2015"
}

@misc{B2023,
	title={Higher representations for extended operators}, 
	author={Thomas Bartsch and Mathew Bullimore and Andrea Grigoletto},
	year={2023},
	eprint={2304.03789},
	archivePrefix={arXiv},
	primaryClass={hep-th},
	url={https://arxiv.org/abs/2304.03789}, 
}

@article{S2018,
   title={Higher-order topological insulators},
   eprint = "1708.03636",
   archivePrefix = "arXiv",
   primaryClass = "math-ph",
   volume={4},
   ISSN={2375-2548},  
   DOI={10.1126/sciadv.aat0346},
   number={6},
   journal={Science Advances},
   publisher={American Association for the Advancement of Science (AAAS)},
   author={Schindler, Frank and Cook, Ashley M. and Vergniory, Maia G. and Wang, Zhijun and Parkin, Stuart S. P. and Bernevig, B. Andrei and Neupert, Titus},
   year={2018},
   month=jun }

@article{Q2024,
  title = {Discovery of Higher-Order Nodal Surface Semimetals},
  author = {Qiu, Huahui and Li, Yuzeng and Zhang, Qicheng and Qiu, Chunyin},
  eprint = "2311.17419",
  archivePrefix = "arXiv",
  primaryClass = "cond-mat.mtrl-sci",
  journal = {Phys. Rev. Lett.},
  volume = {132},
  issue = {18},
  pages = {186601},
  numpages = {7},
  year = {2024},
  month = {Apr},
  publisher = {American Physical Society},
  doi = {10.1103/PhysRevLett.132.186601},
 
}

@article{G1983,
    author = "Green, Michael B. and Schwarz, John H.",
    title = "{Properties of the Covariant Formulation of Superstring Theories}",
    reportNumber = "Print-84-0264 (QUEEN MARY COLL.)",
    doi = "10.1016/0550-3213(84)90030-0",
    journal = "Nucl. Phys. B",
    volume = "243",
    pages = "285--306",
    year = "1984"
}

@article{G1983w,
    author = "Green, Michael B. and Schwarz, John H.",
    title = "{Covariant Description of Superstrings}",
    reportNumber = "QMC-83-7",
    doi = "10.1016/0370-2693(84)92021-5",
    journal = "Phys. Lett. B",
    volume = "136",
    pages = "367--370",
    year = "1984"
}

@article{BS,
author = {Bunk, Severin and S\"{a}mann, Christian and Szabo, Richard J.},
title = {{The 2-Hilbert space of a prequantum bundle gerbe}},
journal = {Reviews in Mathematical Physics},
volume = {30},
number = {01},
pages = {1850001},
year = {2018},
doi = {10.1142/S0129055X18500010},
eprint={1608.08455},
archivePrefix={arXiv},
primaryClass={math-ph}
}

@book{MP,
	author    = {Michor, Peter W.},
	title     = {Topics in Differential Geometry},
	publisher = {American Mathematical Society},
	series    = {Graduate Studies in Mathematics},
	volume    = {93},
	year      = {2008},
	isbn      = {978-0-8218-2003-2},
	url       = {http://catdir.loc.gov/catdir/toc/ecip0813/2008010629.html}
}

@misc{K2024,
	title={Higher condensation theory}, 
	author={Liang Kong and Zhi-Hao Zhang and Jiaheng Zhao and Hao Zheng},
	year={2025},
	eprint={2403.07813},
	archivePrefix={arXiv},
	primaryClass={cond-mat.str-el},
	url={https://arxiv.org/abs/2403.07813}, 
}

@article{K2014436,
	title = {Anyon condensation and tensor categories},
	journal = {Nuclear Physics B},
	volume = {886},
	pages = {436-482},
	year = {2014},
	issn = {0550-3213},
	doi = {https://doi.org/10.1016/j.nuclphysb.2014.07.003},
	url = {https://www.sciencedirect.com/science/article/pii/S0550321314002223},
	author = {Liang Kong}
}

@article{f2023,
	author = {Greg Friedman},
	eprint={0809.4221},
	archivePrefix={arXiv},
	primaryClass={math.AT},
	ISSN = {00357596, 19453795},
	URL = {http://www.jstor.org/stable/44240054},
	journal = {The Rocky Mountain Journal of Mathematics},
	number = {2},
	pages = {353--423},
	publisher = {Rocky Mountain Mathematics Consortium},
	title = {SURVEY ARTICLE: AN ELEMENTARY ILLUSTRATED INTRODUCTION TO SIMPLICIAL SETS},
	urldate = {2026-03-04},
	volume = {42},
	year = {2012}
}

@article{KS2022,
	author = {Kraft, Andreas and Schnitzer, Jonas},
	title = {{An introduction to $L_\infty$-algebras and their homotopy theory for the working mathematician}},
	eprint={2207.01861},
	archivePrefix={arXiv},
	primaryClass={math.QA},
	journal = {Reviews in Mathematical Physics},
	volume = {36},
	number = {01},
	pages = {2330006},
	year = {2024},
	doi = {10.1142/S0129055X23300066}
}

@misc{r2011,
	title={Higher Symplectic Geometry}, 
	author={Christopher L. Rogers},
	year={2011},
	eprint={1106.4068},
	archivePrefix={arXiv},
	primaryClass={math-ph},
         url={https://arxiv.org/abs/1106.4068}, 
}

@article{FMT2024,
	author    = {Freed, Daniel S. and Moore, Gregory W. and Teleman, Constantin},
	title     = {Topological symmetry in quantum field theory},
	journal   = {Quantum Topol.},
	volume    = {15},
	number    = {3/4},
	pages     = {779--869},
	year      = {2024},
	doi       = {10.4171/QT/223},
	eprint={2209.07471},
	archivePrefix={arXiv},
	primaryClass={hep-th},
}

@article{SS2023,
	author = "Sati, Hisham and Schreiber, Urs",
	title = "{Flux Quantization on Phase Space}",
	eprint = "2312.12517",
	archivePrefix = "arXiv",
	primaryClass = "hep-th",
	doi = "10.1007/s00023-024-01438-x",
	journal = "Annales Henri Poincare",
	volume = "26",
	number = "3",
	pages = "895--919",
	year = "2025"
}

@incollection{SS2024,
	title = {Flux Quantization},
	editor = {Richard Szabo and Martin Bojowald},
	booktitle = {Encyclopedia of Mathematical Physics (Second Edition)},
	publisher = {Academic Press},
	edition = {Second Edition},
	address = {Oxford},
	pages = {281-324},
	year = {2025},
	isbn = {978-0-323-95706-9},
	doi = {10.1016/B978-0-323-95703-8.00078-1},
	author = {Hisham Sati and Urs Schreiber},
	eprint = "2402.18473",
	archivePrefix = "arXiv",
	primaryClass = "hep-th"
}

@article{FD2000,
	author = {Freed, Daniel},
	year = {2000},
	month = {12},
	pages = {},
	title = {Dirac Charge Quantization and Generalized Differential Cohomology},
	volume = {7},
	journal = {Surveys in Differential Geometry},
	doi = {10.4310/SDG.2002.v7.n1.a6},
	eprint={hep-th/0011220},
	archivePrefix={arXiv},
	primaryClass={hep-th},
}

@article{KLW,
	title = {Algebraic higher symmetry and categorical symmetry: A holographic and entanglement view of symmetry},
	author = {Kong, Liang and Lan, Tian and Wen, Xiao-Gang and Zhang, Zhi-Hao and Zheng, Hao},
	journal = {Phys. Rev. Res.},
	volume = {2},
	issue = {4},
	pages = {043086},
	numpages = {53},
	year = {2020},
	month = {Oct},
	publisher = {American Physical Society},
	doi = {10.1103/PhysRevResearch.2.043086},
}

@article{McC2013,
author    = {James E. McClure},
title     = {On semisimplicial sets satisfying the {Kan} condition},
journal   = {Homology, Homotopy and Applications},
volume    = {15},
number    = {1},
pages     = {73--82},
year      = {2013},
doi       = {10.4310/HHA.2013.v15.n1.a4},
mrclass   = {55U10},
mrnumber  = {MR3079177},
mrreviewer = {Daniel G. Davis},
zbl       = {1276.55016},
archivePrefix = {arXiv},
eprint    = {1210.5650},
primaryClass = {math.AT}
}

@article{Rou1971,
author    = {C. P. Rourke and B. J. Sanderson},
title     = {{$\Delta$-sets. {I}. {H}omotopy theory}},
journal   = {Quart. J. Math. Oxford Ser. (2)},
volume    = {22},
pages     = {321--338},
year      = {1971},
mrclass   = {55U10},
mrnumber  = {0301784}
}

@article{Wal2021,
author = "Waldorf, Konrad and Kristel, Peter and Ludewig, Matthias",
title = "{2-vector bundles}",
eprint = "2106.12198",
archivePrefix = "arXiv",
primaryClass = "math.DG",
doi = "10.21136/HS.2025.02",
journal = "Higher Struct.",
volume = "9",
number = "1",
pages = "36--87",
year = "2025"
}
\end{document}